\documentclass[sigconf]{acmart}

\renewcommand\footnotetextcopyrightpermission[1]{} 
\makeatletter
\renewcommand\@formatdoi[1]{\ignorespaces}
\makeatother

\usepackage{graphicx}
\usepackage{xcolor,colortbl}
\usepackage{booktabs}
\usepackage{url}
\usepackage{changepage}
\usepackage{amsmath}
\usepackage{amsfonts}
\usepackage{amssymb}
\usepackage{stmaryrd}
\usepackage{syntax}
\usepackage{epsfig}
\usepackage{semantic}
\usepackage{listings}
\usepackage{pifont}
\usepackage{url}
\usepackage{pdflscape}
\usepackage{afterpage}
\usepackage{longtable}
\usepackage{multicol}
\usepackage{mathpartir}
\usepackage{graphicx}
\usepackage{subcaption}

\usepackage{enumitem}
\setlist[enumerate]{noitemsep,topsep=2pt,leftmargin=*}
\setlist[itemize]{noitemsep,topsep=2pt,leftmargin=*}

\usepackage{multirow}

\newtheorem{mydef}{Definition}
\definecolor{Gray}{RGB}{33, 33, 33}
\definecolor{lightlGray}{RGB}{237, 237, 237}

\newcommand{\code}[1]{\texttt{\small #1}}

\usepackage{setspace}

\definecolor{mpcolor}{rgb}{0.1,0.9,0.1}

\definecolor{dscolor}{rgb}{0.1,0.1,0.9}

\definecolor{cascolor}{rgb}{1,0.75,0}

\makeatletter
\newcommand*\mysize{%
  \@setfontsize\mysize{7.7}{9.0}%
}
\makeatother

\usepackage{enumitem}
\setlist[enumerate]{noitemsep,topsep=2pt,leftmargin=*}
\setlist[itemize]{noitemsep,topsep=2pt,leftmargin=*}
\makeatletter
\renewcommand{\paragraph}{%
  \@startsection{paragraph}{4}%
  {\z@}{0.5ex \@plus 0.5ex \@minus .2ex}{-1em}%
  {\normalfont\normalsize\bfseries}%
}
\makeatother

\newenvironment{itemize*}%
{\begin{itemize}}%
  {\end{itemize}}
    
\mathlig{<}{\langle}
\mathlig{>}{\rangle}
\mathlig{|/}{\Downarrow}
\mathlig{|->}{\mapsto}

\newcommand{\node}{Node.js}

\newcounter{eventIdCtr}

\newcommand{\numberapps}{56}

\newcommand{\benchmarksWithHidden}{Five}
\newcommand{\puVsNsuPermissiveness}{5.46}
\newcommand{\dropInLCNIOS}{0.3\%}
\newcommand{\dropInLCOSES}{45.4\%}

\newcommand{\upgrade}{\mathit{upgrade}}

\newcommand{\Obj}{\mathit{Obj}}

\newcommand{\Base}{\mathit{Base}}
\newcommand{\eFlow}{\textbf{E}}
\newcommand{\oFlow}{\textbf{O}}

\newcommand{\zeroFlow}{\mathbf{0}}
\newcommand{\Addr}{\mathit{Addr}}

\newcommand{\Cnt}{\mathit{Cnt}}
\newcommand{\SBC}{\mathit{SBC}}
\newcommand{\LCR}{\mathit{LCR}}
\newcommand{\Conf}{\mathit{Conf}}
\newcommand{\Cxt}{\mathit{Cxt}}

\newcommand{\pemph}[1]{\textbf{#1}}
\newcommand{\Stmt}{\mathit{Stmt}}

\newcommand{\psink}[1]{\pemph{sink}(#1)}
\newcommand{\pskip}{\pemph{skip}}
\newcommand{\terminated}{\varepsilon}
\newcommand{\emptyTrace}{\varepsilon}
\newcommand{\ptrue}{\pemph{tt}}
\newcommand{\pfalse}{\pemph{ff}}
\newcommand{\passign}[2]{#1 = #2}
\newcommand{\passignArr}[3]{#1[#2] = #3}
\newcommand{\passignField}[3]{#1.#2 = #3}
\newcommand{\dom}{\text{dom}}
\newcommand{\ppop}{\textbf{pop}}
\newcommand{\pif}[3]{\pemph{if}\ (#1)\ \{\ #2\ \}\ \pemph{else}\ \{\ #3\ \} }

\newcommand{\pwhile}[2]{\pemph{while}\ #1\ \pemph{do}\ #2}
\newcommand{\pseq}[2]{#1 \mathrel{;} #2}

\newcommand{\pupgrade}[1]{\pemph{upgrade}(#1)}

\newcommand{\High}{\textbf{H}}
\newcommand{\Name}{\mathit{Name}}
\mathchardef\breakingcomma\mathcode`\,
{\catcode`,=\active
  \gdef,{\breakingcomma\discretionary{}{}{}}
}
\newcommand{\mathlist}[1]{\mathcode`\,=\string"8000 #1}
\newcommand{\conf}[1]{\langle #1 \rangle}
\newcommand{\Label}{\mathcal{L}}
\newcommand{\Value}{\mathit{Value}}

\newcommand{\Trace}{\mathit{Tr}}
\newcommand{\Expr}{\mathit{Expr}}

\newcommand{\expeval}[2]{\llbracket #1 \rrbracket (#2)}
\newcommand{\env}{\rho}
\newcommand{\cnt}{\kappa}
\newcommand{\cnf}{\mathit{cnf}}
\newcommand{\expl}{\mathit{expl}}
\newcommand{\obs}{\mathit{obs}}
\newcommand{\step}[1]{\xrightarrow{#1}}
\newcommand{\flowsto}{\sqsubseteq}
\newcommand{\cxtInsert}[2]{#1[#2]}
\newcommand{\leaveBranch}[1]{\mathit{leaveBranch}(#1)}
\newcommand{\length}[1]{\mathit{length}(#1)}
\newcommand{\toVal}{\mathit{toVal}}
\allowdisplaybreaks

\begin{document}
%
%
%
\title{An Empirical Study of Information Flows in Real-World JavaScript}

\author{Cristian-Alexandru Staicu}
\affiliation{TU Darmstadt}
\author{Daniel Schoepe}
\affiliation{Chalmers University of Technology}
\author{Musard Balliu}
\affiliation{KTH Royal Institute of Technology}
\author{Michael Pradel}
\affiliation{TU Darmstadt}
\author{Andrei Sabelfeld}
\affiliation{Chalmers University of Technology}


%
%
%

\lstdefinelanguage{JavaScript}{
  keywords={typeof, new, true, false, catch, function, return, null, catch, switch, var, if, in, while, do, else, case, break},
  keywordstyle=\color{darkgray}\bfseries,
  ndkeywords={class, export, boolean, throw, implements, import, this},
  ndkeywordstyle=\color{darkgray}\bfseries,
  identifierstyle=\color{black},
  sensitive=false,
  comment=[l]{//},
  morecomment=[s]{/*}{*/},
  commentstyle=\color{darkgray}\ttfamily,
  stringstyle=\color{gray}\ttfamily,
  morestring=[b]',
  morestring=[b]"
}

\lstset{
   language=JavaScript,
   extendedchars=true,
   basicstyle=\mysize\ttfamily,
   showstringspaces=false,
   showspaces=false,
   numbers=left,
   numberstyle=\scriptsize,
   numbersep=9pt,
   tabsize=2,
   breaklines=true,
   showtabs=false,
   captionpos=b,
   framesep=4.5mm,
   framexleftmargin=2.5mm,
   escapeinside={/*\#}{\#*/},
   xleftmargin=2em
}

\makeatletter
\newcommand\HL{%
   \gdef\lst@alloverstyle##1{%
     \textcolor{red}{##1}
   }%
}

\newcommand\HLoff{%
   \xdef\lst@alloverstyle##1{##1}%
}

\newcommand\HLL{%
   \gdef\lst@alloverstyle##1{%
     \textcolor{blue}{##1}
   }%
}

\newcommand\HLLoff{%
   \xdef\lst@alloverstyle##1{##1}%
}
\setcopyright{none}

\begin{abstract}
Information flow analysis 
prevents secret or untrusted data from flowing into public or trusted 
sinks.
Existing mechanisms cover a wide array of options, ranging from lightweight 
taint analysis to heavyweight information flow control that also considers 
implicit flows.
Dynamic analysis, which is particularly popular for languages such as 
JavaScript, faces the question whether to invest in analyzing flows 
caused by not executing a particular branch, so-called hidden implicit 
flows.
This paper addresses the questions how common different kinds of flows are 
in real-world programs, how important these flows are to enforce security 
policies, and how costly it is to consider these flows.
We address these questions in an empirical study that analyzes \numberapps{} 
real-world JavaScript programs that suffer from various security problems, 
such as code injection vulnerabilities, denial of service vulnerabilities,
memory leaks, and privacy leaks.
The study is based on a state-of-the-art dynamic information flow analysis 
and a formalization of its core.
We find that implicit flows are expensive to track in terms of 
permissiveness, label creep, and runtime overhead.
We find a lightweight taint analysis to be sufficient for most of 
the studied security problems, while for some privacy-related code, 
observable tracking is sometimes required.
In contrast, we do not find any evidence that tracking hidden implicit flows 
reveals otherwise missed security problems.
Our results help security analysts and analysis designers to understand the 
cost-benefit tradeoffs of information flow analysis and provide empirical 
evidence that analyzing implicit flows in a cost-effective way is a relevant 
problem.

\end{abstract}

%
%
%
%
%

\maketitle
\section{Introduction}
\label{sec:intro}

JavaScript is at the heart of the modern web, empowering rich client-side 
applications and, more recently, also server-side applications.
While some language features, such as dynamism and flexibility, 
explain this popularity, the lack of other features, such as language-level 
protection and isolation mechanisms, open up a
wide range of integrity, availability, and confidentiality vulnerabilities~\cite{DBLP:journals/virology/Johns08}.
As a result, securing JavaScript applications has become
a key challenge for web application security.
Unfortunately, existing browser-level mechanisms, such as the same-origin 
policy or the content security policy, are
coarse-grained, falling short to distinguish between secure and
insecure manipulation of data by scripts.
Furthermore, 
server-side applications lack such isolation mechanisms completely, 
allowing an attacker, e.g., to inject and execute arbitrary code that 
interacts with the operating system through powerful 
APIs~\cite{ndss2018}.

An appealing approach to securing JavaScript applications is information 
flow analysis.
This approach tracks the flow of information from sources to sinks in
order to 
enforce application-level security policies.
It can ensure both integrity, by preventing information from untrusted 
sources to reach trusted sinks, and confidentiality, by preventing 
information from secret sources to reach public sinks.
For example, information flow analysis can check that no attacker-controlled 
data is evaluated as executable code or that secret user data is not sent to 
the network.
Because the dynamic nature of JavaScript hinders precise static analysis, 
dynamic information flow analysis has received significant attention by 
researchers~\cite{DBLP:conf/ndss/VogtNJKKV07,Jang2010,Chudnov:2015:IIF:2810103.2813684,hedin2014jsflow,DBLP:conf/esorics/BichhawatRJGH17,bauer2015run,DeGroef:2012:FWB:2382196.2382275,Austin:2012:MFD:2103656.2103677}.
The basic idea of dynamic information flow analysis is to attach security 
labels, e.g., secret (untrusted) and public (trusted), to runtime values and 
to propagate these labels during program execution.
To simplify the presentation, we assume to have two security labels, and we 
say that a value is \emph{sensitive} if its label is secret or untrusted;
otherwise, we say that a value is \emph{insensitive}.

At the language level, a program may propagate information via two kinds of 
information flows:\footnote{There are other kinds of flows, such as timing 
and cache side-channels, which we ignore here.}
\emph{Explicit flows}~\cite{denning1977certification} occur whenever sensitive information is passed
by an assignment statement or into a sink.
\emph{Implicit flows}~\cite{denning1977certification} arise via control-flow structures of programs, e.g., 
conditionals and loops, when the flow of control depends on a sensitive value.
For a dynamic information flow analysis, implicit flows can be further 
classified into flows that happen because a particular branch is executed, 
so-called \emph{observable implicit flows}~\cite{DBLP:conf/esorics/BalliuSS17}, and flows that happen
because a particular branch is not executed, so called \emph{hidden 
implicit flows}~\cite{DBLP:conf/esorics/BalliuSS17}.



\begin{figure}
\begin{spacing}{0.9}
\begin{lstlisting}
// variable /*#\HL#*/passwd/*#\HLoff#*//*#\label{line:markSource}#*/is sensitive 
var gotIt = false;
var paddedPasswd = "xx" + passwd;/*#\label{line:explicit}#*/
var knownPasswd = null;
if (paddedPasswd === "xxtopSecret") {/*#\label{line:cond}#*/
  gotIt = true;/*#\label{line:write1}#*/
  knownPasswd = passwd;/*#\label{line:write2}#*/
}/*#\label{line:merge}#*/
// function /*#\HLL#*/sink/*#\HLLoff#*/is insensitive
sink(gotIt);/*#\label{line:sink}#*/
\end{lstlisting}
\end{spacing}
\caption{Program leaking the password to the network.} 
\label{fig:example}
\end{figure}

Figure~\ref{fig:example} illustrates the different kinds of flows with a 
simple JavaScript-like program that leaks sensitive information.
The program has a variable \code{passwd}, which is marked initially as a sensitive source at 
line~\ref{line:markSource}.
Using this variable in an operation that creates a new value, e.g., in 
line~\ref{line:explicit}, is an explicit flow.
Consider the case where the password is ``topSecret'', i.e., the conditional 
at line~\ref{line:cond} evaluates to \code{true}, and line~\ref{line:write1} 
sets \code{gotIt} to \code{true}.
At line~\ref{line:sink}, the \code{gotIt} variable is sent to the network 
through the function \code{sink()}, which is considered to be an insensitive sink.  
The flow from the password to \code{gotIt} is an observable implicit flow 
because a sensitive value determines that \code{gotIt} gets written.
Now, consider the case where \code{passwd} is ``abc''. The
branch at line~\ref{line:cond} is not taken and the \code{gotIt} variable 
remains \code{false}. Sending this information to the network reveals that 
the password is different from ``topSecret''. This flow is a hidden implicit 
flow because a sensitive value determines that \code{gotIt} does not get 
written.

Ideally, an information flow analysis should consider all three kinds of 
flows.
In fact, there exists a large body of work on static, dynamic, hybrid, and 
multi-execution techniques to prevent explicit and implicit flows.
However, so far these tools have seen little use in practice, despite the 
strong security guarantees that they provide.
In contrast, a lightweight form of information flow analysis called 
\emph{taint analysis} is widely used in computer security~\cite{DBLP:conf/sp/SchwartzAB10}.
Taint analysis is a pure data dependency analysis that only tracks explicit 
flows, ignoring any control flow dependencies. 

The question which kinds of flows to consider is a tradeoff between costs 
and benefits.
On the cost side, considering more flows increases false positives~\cite{King2008}.
A false positive here means that a secure execution is conservatively
blocked by an overly restrictive enforcement mechanism. 
A common reason is that a value gets labeled as sensitive even though it does
not actually contain information that is security-relevant in practice.
This problem, sometimes referred to as
\emph{label creep}~\cite{RoblingDenning:1982:CDS:539308,Sabelfeld2003}, reduces the 
permissiveness of information flow monitoring, because the monitor will 
prematurely stop a program to prevent a value 
with an overly sensitive label from reaching a sink.
Another cost of considering more kinds of flows is an increase in runtime 
overhead.
On the benefit side, considering more flows increases the ability to find 
security vulnerabilities and data leakages, i.e., the level of trust one 
obtains from the analysis.
For example, an analysis that considers only explicit flows will miss any 
leakage of sensitive data that involves an implicit flow.
Unfortunately, despite the large volume of research on information flow 
analysis, there is very little empirical evidence on the importance of the 
different kinds of flows in real applications.
Because of this lack of knowledge, potential users of information flow 
analyses cannot make an informed decision about what kind of analysis to
use.


To better understand the tradeoff between costs and benefits of using a 
dynamic information flow analysis, this paper presents an empirical study of 
information flows in real-world JavaScript code.
Our overall goal is to better understand the costs and benefits of 
dynamically analyzing explicit, observable implicit, and hidden implicit 
flows.
Specifically, we are interested in how prevalent different kinds of flows 
are, what kinds of security problems can(not) be detected when considering 
subsets of flows, and what costs considering all flows imposes.
%
To address these questions, we study \numberapps{} real-world JavaScript 
programs in various application domains with a diverse set of security 
policies.
The study considers integrity problems, specifically code injection 
vulnerabilities and denial of service vulnerabilities caused by an 
algorithmic complexity problem,
and confidentiality problems, specifically leakages of uninitialized memory, 
browser fingerprinting and history sniffing.
Each studied program has at least one real-world security problem that 
information flow analysis can detect.

Our study is enabled by a novel methodology that combines state-of-the-art 
dynamic information flow  
analysis~\cite{Hedin2012,hedin2014jsflow,austin2010permissive} and program 
rewriting~\cite{birgisson2012boosting} with a set of novel \emph{security 
metrics}.
We implement the methodology in a dynamic information flow analysis built on 
top of Jalangi~\cite{Sen2013}.
The implementation draws on a sound analysis for a simple core of 
JavaScript.
The formalization relates the security metrics to semantic security 
conditions for taint tracking~\cite{explicitsecrecy}, observable tracking~\cite{DBLP:conf/esorics/BalliuSS17} and 
information flow monitoring~\cite{DBLP:dblpconf/sp/GoguenM82}.

The findings of our study include:
\begin{enumerate}
\item
    All three kinds of flows occur locally in real-life applications, i.e., an 
    analysis that 
    ignores some of them risks to miss violations of the information flow 
    policy. Explicit flows are by far the most prevalent, and only five benchmarks contain hidden implicit flows (Section~\ref{subsec:mf}).  
  \item
    An analysis that considers explicit and observable implicit flows, but 
    ignores hidden implicit flows, detects all vulnerabilities in our 
    benchmarks.
    For most applications it is even sufficient to track explicit flows 
    only, while for some client-side, privacy-related applications one must 
    also consider observable implicit flows (Section~\ref{subsec:s2s}). 
   
  \item
  Tracking hidden implicit flows causes an analysis to prematurely terminate 
  various executions. Furthermore, we find that different monitoring 
  strategies proposed in the literature vary significantly in their 
  permissiveness.  (Section~\ref{perm}).
  
   \item
  The amount of data labeled as sensitive steadily increases during the 
  execution of most benchmarks, confirming the label creep problem.
  An analysis that considers implicit flows increases the label creep by 
  over 40\% compared to an analysis that considers only explicit flows 
  (Section~\ref{subsec:lcr}).
    
 \item
    The analysis overhead caused  by considering implicit flows is 
    significant: Ignoring implicit flows saves the effort of tracking 
    runtime operations by a factor of 2.5 times (Section~\ref{runtime}).
\end{enumerate}

Prior work (discussed in Section~\ref{sec:relatedwork}) studies false 
positives caused by static analysis of implicit 
flows~\cite{King2008,Russo+:MOD09} and the semantic strength of 
flows~\cite{masri2009measuring}. Jang et al.~\cite{Jang+:CCS10} conduct a 
large-scale empirical study showing that
several popular web sites use information flows to exfiltrate data about users' behavior. Kang et al.~\cite{DBLP:conf/ndss/KangMPS11}
combine dynamic taint analysis with targeted implicit flow analysis, demonstrating the importance of tracking implicit flows for trusted
programs.
However, to the best of our knowledge, no 
existing work addresses the above questions.

In summary, this paper contributes the following:
\begin{itemize*}
  \item We are the first to empirically study the prevalence of explicit, 
    observable implicit, and hidden implicit flows in real-world 
    applications against integrity, availability, and confidentiality policies.
  \item We present a methodology and its implementation, which enables the study, and we provide a 
  formal basis for empirically studying information flows (Section~\ref{sec:methodology}).
  \item We show the soundness of the analysis for a core of JavaScript with 
    respect to semantic security conditions (Appendix).
  \item Through realistic case studies and security policies, we provide 
    empirical evidence that sheds light on the cost-benefit tradeoff of 
    information analysis and that outlines directions for future work (Section~\ref{sec:results}).
\end{itemize*}

We share our implementation, as well as all benchmarks and policies used for 
the study, to support future evaluations of information flow tools for 
JavaScript.\footnote{\tiny\url{https://new-iflow.herokuapp.com/download-iflow.html}} 

\section{Benchmarks and Security Policies}
\label{sec:background}

Our study is based on \numberapps{} client-side and server-side JavaScript applications, which suffer from four classes of vulnerabilities.
These applications are subject to  attacks that have been independently discovered by existing work, including 
integrity, availability, and confidentiality attacks. For every application, 
we define realistic security policies expressed as information flow 
policies.
Table~\ref{tab:apps} shows the applications, along with their security policies, and size measured in lines of code.
The benchmarks vary in size from tens of lines of code to tens of thousands.
We further explain the policies below.
For each application we either create or reuse a set of inputs that trigger the attack and other inputs to increase the coverage of different behaviors.



\begin{table}
	\caption{Insecure programs, security policies, program size, sensitive
		branch coverage and  number of upgrades. "module" stands for the module interface.}
	\centering
	\setlength{\tabcolsep}{.8pt}
	\footnotesize
	\def\arraystretch{1}%
	\begin{tabular}{@{}rp{2.3cm}p{4.2cm}rrc@{}}
		\toprule
		ID & Library & Policy & LoC & SBC &Upgs \\
		\midrule
		1 & fish  						& module $\rightarrow$ eval and
		exec
		& 69	& 1 & 0 \\
		2 & growl  						& module $\rightarrow$ eval and
		exec
		& 270 & 1 & 0  \\
		3 & gm  							& module $\rightarrow$ eval and
		exec
		& 1,614 & 1 & 0\\
		4 & libnotify						& module $\rightarrow$ eval and
		exec
		& 54 & 1& 0 \\
		5 & mixin-pro						& module $\rightarrow$ eval and
		exec
		& 168 & 1& 0  \\
		6 & modulify						& module $\rightarrow$ eval and
		exec
		& 2,410 &1& 0 \\
		7 & mol-proto						& module $\rightarrow$ eval and
		exec
		& 1,696 & 1& 0\\
		8 & mongoosify					& module $\rightarrow$ eval and
		exec
		& 160 & 0& 1 \\
		9 & m-log							& module $\rightarrow$ eval and
		exec
		& 243  & 1& 0	 \\
		10 & mobile-icon-resizer			& file system API $\rightarrow$
		eval and exec
		& 410  & 1 & 0 \\
		11 & mongo-parse					& module $\rightarrow$ eval and
		exec
		& 506 & 1 & 0 \\
		12 & mongoosemask					& module $\rightarrow$ eval and
		exec   &
		12,750  & 0.78&
		28 \\
		13 & mongui						& HTTP API $\rightarrow$ eval and
		exec
		& 1,539 & 0.44& 0
		\\
		14 & mongo-edit					& HTTP API $\rightarrow$ eval and
		exec
		&  577  & 0& 0 \\
		15 & mock2easy						& HTTP API $\rightarrow$ eval and
		exec   & 1,217  &
		0.07& 3\\
		16 & chook-growl-reporter			& module $\rightarrow$ eval and
		exec   &
		243  &
		1& 0\\
		17 & git2json						& module $\rightarrow$ eval and
		exec   &
		434  &
		1& 0\\
		18 & kerb\_request					& module $\rightarrow$ eval
		and
		exec
		& 67 &
		1& 0\\
		19 & printer						& module $\rightarrow$ eval and
		exec   &
		139 &
		1& 0\\
		\midrule
		20 & debug						& module $\rightarrow$ regex
		matching  &
		360 & 1&
		0\\
		21 & mime						& module $\rightarrow$ regex
		matching  &
		108 & 1&
		0\\
		22 & tough-cookie				& module $\rightarrow$ regex
		matching  &
		1,145 & 1&
		0\\
		23 & fresh						& module $\rightarrow$ regex
		matching  &
		59 & 0.5&
		0\\
		24 & forwarded					& module $\rightarrow$ regex
		matching  &
		30 &
		0& 0\\
		25 & underscore.string			& module $\rightarrow$ regex
		matching  &
		1,779 &
		1& 0\\
		26 & ua-parser-js				& module $\rightarrow$ regex
		matching  &
		584 & 0.50&
		6\\
		27 & parsejson					& module $\rightarrow$
		regex matching  & 46 & 1& 0\\
		28 & useragent					& module $\rightarrow$
		regex matching  & 6,827 & 1& 0\\
		29 & no-case					& module $\rightarrow$
		regex matching  & 33 & 1& 0\\
		30 & content-type-parser		& module $\rightarrow$ regex
		matching  &
		221 & 1&
		0\\
		31 & timespan					& module $\rightarrow$ regex
		matching  &
		577 & 0.20&
		4\\
		32 & string					& module $\rightarrow$ regex
		matching  &
		2,001 & 1&
		0\\
		33 & content					& module $\rightarrow$
		regex matching  & 125 & 0.42& 0\\
		34 & slug						& module $\rightarrow$
		regex matching  & 375 & 0.5& 2\\
		35 & htmlparser				& module $\rightarrow$ regex
		matching  &
		2,155 & 0.65&
		5\\
		36 & charset					& module $\rightarrow$
		regex matching  & 49 & 0.5& 0\\
		37 & mobile-detect				& module $\rightarrow$
		regex matching  & 612 & 1& 0\\
		38 & ismobilejs				& module $\rightarrow$ regex
		matching  & 935
		& 0.33&
		1\\
		39 & dns-sync					& module $\rightarrow$ regex
		matching  &
		76 &
		1& 0\\
		\midrule
		40 & ip						& buffer reading  $\rightarrow$ module &
		325 & 0.76& 0\\
		41 & concat-stream				& buffer reading  $\rightarrow$
		module & 132 & 1& 0\\
		42 & bl						& buffer reading  $\rightarrow$ module &
		206 & 0.72& 4\\
		43 & request					& buffer reading  $\rightarrow$ HTTP
		& 2,217 &
		0.52&
		0\\
		44 & ws						& buffer reading  $\rightarrow$ HTTP API &
		2,449 & 0.07&
		1\\
		45 & floody					& buffer reading  $\rightarrow$ HTTP API &
		94 &
		0.8& 0\\
		46 & tunnel-agent				& buffer reading  $\rightarrow$ HTTP
		API & 225 &
		1& 0\\

		\midrule

		47 & History sniffing~\cite{Jang2010}	& HTMLElement.color  $\rightarrow$
		img.src & 42  	&0& 3  \\
		48 & Font fingerpr.~\cite{acar2013fpdetective}	& HTMLElement.offsetWidth
		$\rightarrow$ img.src & 145	&0.5& 1  \\
		49 & Font fingerpr.\footnotemark{}
		& HTMLElement.offsetWidth
		$\rightarrow$ img.src & 44  	&0.02& 3  \\
		50 & Font
		fingerpr.\footnotemark{}
		& HTMLElement.offsetWidth
		$\rightarrow$ img.src & 134  	&1& 0  \\
		51 & Browser ext.\ fingerpr.~\cite{DBLP:conf/codaspy/SjostenAS17}
		& HTMLElement.offsetWidth
		$\rightarrow$ request.open & 1,451  	&1& 1  \\
		52 & DoNotTrack
		leakage\footnotemark{}
		& navigator\_doNotTrack
		$\rightarrow$ HTMLElement.html & 20	&0& 1  \\
		53 & Login state
		leakage\footnotemark{}
		& onload event
		$\rightarrow$ document.innerHTML & 191  	&1& 0  \\
		54 & Engine
		fingerpr.\footnotemark{}
		& HTMLElement.type
		$\rightarrow$ console.log &  129	&0& 1  \\
		55 & Browser ext.\ fingerpr.\footnotemark{}
		& onload event
		$\rightarrow$ HTMLElement.innerHTML & 37  	&0& 0  \\
		56 & Resource
		fingerpr.\footnotemark{}
		& onload event
		$\rightarrow$ console.log & 43 	&0& 0  \\

		\bottomrule
	\end{tabular}
	\label{tab:apps}
\end{table}

\addtocounter{footnote}{-6}
\footnotetext{\tiny 
	\url{https://www.privacytool.org/AnonymityChecker/}}
\addtocounter{footnote}{1}
\footnotetext{\tiny\url{http://www.lalit.org/lab/javascript-css-font-detect/}}
\addtocounter{footnote}{1}
\footnotetext{\tiny\url{https://browserleaks.com/js/donottrack.js}}
\addtocounter{footnote}{1}
\footnotetext{\tiny\url{https://robinlinus.github.io/socialmedia-leak/}}
\addtocounter{footnote}{1}
\footnotetext{\tiny\url{https://www.privacytool.org/AnonymityChecker/}}
\addtocounter{footnote}{1}
\footnotetext{\tiny\url{https://popmyads.com/}}
\addtocounter{footnote}{1}
\footnotetext{\tiny\url{https://browserleaks.com/firefox\#more}}

Our goal is an \emph{in-depth} study of the different kinds of information 
flows for a range of security policies; we do not claim to study a representative 
sample of JavaScript applications.
Existing \emph{in-breadth} empirical studies, which analyze hundreds of 
thousands of web pages against fixed policies, provide clear evidence for 
security and privacy risks in JavaScript 
code~\cite{Jang2010,Lekies2013,domxss:ndss18}. In contrast to these 
large-scale studies, our effort
consists in identifying vulnerable scripts from different domains and 
analyzing the flows therein.

\paragraph{Injection vulnerabilities on \node{}} The \node{} 
ecosystem has enabled a proliferation of  server and desktop applications 
written in JavaScript.
Injection vulnerabilities are programming errors that enable an attacker to 
inject and execute malicious code.
Recent work~\cite{ndss2018} has demonstrated the devastating impact of 
injection vulnerabilities on server-side programs, e.g., when
an attacker-controlled string reaches powerful APIs such as \code{exec} or 
\code{eval}.
Such attacks can severely compromise integrity, e.g.,
deleting all files in a directory or completely controlling the attacked 
machine.
%
We study 19 \node{} modules that contain injection vulnerabilities (IDs 1 to 19 in Table~\ref{tab:apps}).  As 
security policies, we consider the interface of a module as an untrusted 
source and the APIs that interpret strings as code, such as
\code{exec} or \code{eval}, as trusted sinks.

\paragraph{ReDoS vulnerabilities} Regular expression Denial of Service, or 
ReDoS, is a form of algorithmic complexity attack that exploits the possibly 
long time of matching a regular expression against an attacker-crafted 
input.
The single-threaded execution model of JavaScript makes JavaScript-based web 
servers particularly susceptible to ReDoS attacks~\cite{usenixSec2018}.
We analyze 19 web server applications that are subject to ReDoS attacks (IDs 20 to 39 in Table~\ref{tab:apps}).
As a security policy, we consider data received via HTTP requests as 
untrusted sources and regular expressions known to be vulnerable as trusted 
sinks.

\paragraph{Buffer vulnerabilities}
Buffer vulnerabilities expose memory content filled with previously used
data, e.g., cryptographic keys, source code, or system information.
In \node{}, such vulnerabilities occur when using the \code{Buffer} 
constructor without explicit initialization.
Buffer vulnerabilities are similar to the infamous Heartbleed flaw in 
OpenSSL~\cite{heartbleed}, as both allow an attacker to read more memory 
than intended.
We analyze 7 applications subject to buffer vulnerabilities (IDs 40 to 46 in Table~\ref{tab:apps}).
The security policy requires that no information flows from the buffer 
allocation constructor to HTTP requests without initialization.

\paragraph{Device fingerprinting and history sniffing}
Web-based fingerprinting collects device-specific information, e.g., 
installed fonts or browser extensions, to identify users~\cite{Acar2014}.
History sniffing attacks use the fact that browsers display links 
differently depending on whether the target has been 
visited~\cite{Jang2010,weinberg2011still}.
We analyze 10 client-side JavaScript applications that are subject to 
various forms of fingerprinting and history sniffing attacks (IDs 47 to 56 in Table~\ref{tab:apps}).
The security policies label as secret the sources that provide sensitive 
information, e.g., the font height and width, and as public sinks the APIs 
that enable external communication, e.g., image tags.
We adapt these programs to our Node.js-based infrastructure by introducing
minimal changes that emulate DOM interactions. We carefully cross-checked 
this adaptations in a pair-programming fashion, ensuring that all flows in 
the original program are preserved. The policies are application-specific 
and mark certain nodes in the emulated DOM as sources and sinks. In contrast 
to the other benchmarks, these programs can potentially be 
malicious~\cite{DBLP:conf/sp/NikiforakisKJKPV13,Jang2010}.  That is, the 
assumption that the analyzed code is trusted  does no longer hold.

\section{Methodology}
\label{sec:methodology}

To address the research questions from Section~\ref{sec:intro}, we present a 
methodology that combines a set of novel metrics with a dynamic information 
flow analysis~\cite{Hedin2012,hedin2014jsflow}, a monitoring 
strategy~\cite{austin2010permissive}, and an automated mechanism to insert
upgrade statements~\cite{birgisson2012boosting}.
The metrics summarize the flows observed during the program execution. 
This section provides the necessary background on information flow analysis, 
an informal description of our methodology, and definitions of the  
metrics. It also presents a formalization of the core of our 
methodology.

\subsection{Setting: Information Flow Analysis}
\label{sec:background1}

\begin{table}[tb]
  \setlength{\tabcolsep}{1.3pt}
  \renewcommand{\arraystretch}{1.1}
  \centering
  \footnotesize
  \caption{Monitoring strategies (``Expl.'' = explicit, ``Obs.'' 
  = observable implicit, ``Hid.'' = hidden implicit).}
  \begin{tabular}{@{}p{9em}p{8.2em}cccp{10em}@{}}
    \toprule

Strategy & Sec. condition & \multicolumn{3}{c}{Tracked flows} & 
Permissiveness \\
\cmidrule{3-5}
&& Expl. & Obs. & Hid. \\

\midrule

Taint tracking & Explicit \mbox{secrecy} & \checkmark &&&
Stop when $H$-labeled value reaches sink. \\

Observable tracking & Observable secrecy & \checkmark & \checkmark &&
Stop when $H$-labeled value reaches sink. \\

No Sensitive Upgrade &
Non-interference &
\checkmark &
\checkmark &
\checkmark &
Stop when $L$-labeled variable is written in sensitive context.\\

Permissive Upgrade &
Non-interference &
\checkmark &
\checkmark &
\checkmark &
Stop when partially leaked value is used.\\

    \bottomrule
  \end{tabular}
  \label{tab:strategies}
\end{table}

\paragraph*{Security labels}
An information flow analysis associates each value with a \emph{security 
label} that indicates how sensitive the value is. Labels are typically 
arranged in a lattice~\cite{Denning1976}.
To ease the presentation, we focus on two labels: $H$ (for high or sensitive) 
and $L$ (for low or insensitive), where $H$ is more sensitive than $L$.
Given a label $\ell\in \{H,L\}$, we write $v^\ell$ to denote that a value 
$v$ has security label $\ell$. If a value $v$ does not have a label, we 
assume it is implicitly labeled as $L$.

\paragraph*{Information flow policy}
The analysis checks whether data from a sensitive \emph{source} influences 
data that arrives at an insensitive \emph{sink}.
The sources and sinks for a program are specified in an \emph{information 
flow policy}, or short, \emph{policy}.
For integrity, the policy specifies that no information from untrusted 
sources ($H$) reaches trusted sinks ($L$).
For confidentiality, the policy stipulates that no information from secret 
sources ($H$) reaches public sinks ($L$).
We model sources by variables and object fields, and their security label 
corresponds to the label of the value that they contain initially.
We denote sinks by a function \code{sink()}, which is implicitly 
labeled as $L$.

\paragraph*{Monitoring strategies}
Different \emph{monitoring strategies} for dynamic information flow analysis 
address the problem of checking whether an execution violates a policy.
In this work, we focus on flow-sensitive dynamic monitors, where variables 
can be assigned  different security labels during the execution.
Table~\ref{tab:strategies} gives an overview of the monitoring strategies 
studied in this paper. 
Taint analysis tracks only explicit flows and stops the program only if an 
$H$-labeled value reaches a sink.

In contrast to taint tracking, the other two strategies also track implicit 
flows.
The monitors identify implicit flows by maintaining a \emph{security stack} 
that contains all sensitive labels of expressions in conditionals that influence the control flow.
When the stack is non-empty, the program executes in a \emph{sensitive 
context}.
Observable Tracking~\cite{DBLP:conf/esorics/BalliuSS17} tracks only explicit 
and observable implicit flows, but ignores hidden implicit flows.
Whenever an $L$-labeled variable is updated in a sensitive context, 
observable tracking updates the label as sensitive and continues with the 
execution.
For example, consider the following program, which is trivially secure 
because there is no call to \code{sink()}:

\begin{lstlisting}
var location; var y; var z;
if (10 < location < 20)/*#\label{line:cond1}#*/ {
  y = "Home";/*#\label{line:write3}#*/ }  
//upgrade(y);/*#\label{line:write5}#*/  
z = "You are at " + y;/*#\label{line:write4}#*/         
\end{lstlisting}

Consider now an execution where the location is $15^H$.
Observable tracking updates the labels of \code{y} and \code{z} as sensitive 
and does not stop the execution.

The strictest monitoring strategies try to prevent also hidden implicit
flows.
We consider two variants of such a strategy.
They both terminate the execution of the program whenever an observable 
implicit flow may lead to a hidden implicit flow in another execution.
The No Sensitive Upgrade strategy (NSU)~\cite{Zdancewic:2002:PLI:935787,austin2009efficient}
disallows updating the security labels of a variable in a sensitive context.
In particular, it terminates the execution whenever such an update happens.
For example, consider the execution of the above program with 
\code{location=15$^H$}.
The NSU strategy terminates the program at line~\ref{line:write3} due to 
the update of the $L$-labeled variable \code{y} in a sensitive context.

Permissive Upgrade (PU)~\cite{austin2010permissive}  is a refinement of the 
NSU strategy.
It labels a value as \emph{partially leaked} if an $L$-labeled variable is 
updated in a sensitive context, and terminates the program if the updated 
variable is further used outside the sensitive context.
Consider again the same execution of the above program.
The PU strategy labels \code{y} as partially leaked at 
line~\ref{line:write3} because the program writes to the $L$-labeled 
variable in a sensitive context, and then terminates the program at 
line~\ref{line:write4} because the value is used.
In our work, we use the PU strategy to study the prevalence of different 
kinds of flows.

\paragraph*{Upgrade statements}
Naively applying the PU strategy to real-world programs can be very 
restrictive and risks to increase the number of false positives, i.e., 
terminate many secure executions. To address this problem, Austin and 
Flanagan propose the \emph{upgrade statement}~\cite{austin2009efficient} and 
the \emph{privatization statement}~\cite{austin2010permissive}.  These 
statements change the label of a variable to $H$ explicitly, to signal a 
potential hidden implicit flow to the monitor.  %
For example,  we can insert an upgrade statement before 
line~\ref{line:write4} in the above example to mark \code{y} as sensitive 
even if the branch is not taken.
As a result, the program does not terminate immediately when the value is 
read.
If the program would later call \code{sink(y)}, then the monitor would 
terminate the program and report a policy violation.

%


\paragraph*{Permissiveness}
The above example illustrates the permissiveness issues of different 
monitoring strategies, i.e., that they terminate the program unnecessarily 
even though no policy violation occurs.
Taint tracking and observable tracking both do not terminate the program.
In contrast, both NSU and PU terminate the program unnecessarily.
This overapproximation of policy violations is necessary to avoid
potential hidden implicit flows.
Adding upgrade statements avoids such premature termination of the program 
by assigning an $H$-label to \code{y}, independently of what branch of the 
conditional statement is executed.
If we uncomment line~\ref{line:write5}, the execution proceeds without 
terminating the program unnecessarily.
That is, upgrade statements may increase the permissiveness, but impose the 
cost of adding upgrade statements.

\subsection{Security Metrics}

Our approach uses program testing to measure the prevalence of different 
kinds of information flows. The basic idea is to test a program with an 
information flow monitor that implements the PU strategy, while incrementing 
counters that represent the number of explicit, observable implicit, and 
hidden implicit flows.
These counters then allow us to reason about the prevalence of the different 
kinds of flows and about the policy violations that different monitoring 
strategies would detect.
In contrast to the PU monitor that terminates the program when it encounters 
a policy violation, our monitor continues the execution to measure flows in 
the remainder of the execution. We refer to Section~\ref{sec:secFram}
for the formal definition of the monitor.

We consider information flows at two levels of granularity.
On the one hand, we consider flows induced by a single operation in the program 
(Section~\ref{sec:metricsMicroFlows}).  We call such flows \emph{micro 
flows} or simply flows. Studying flows at the micro flow level is worthwhile because it 
provides a detailed understanding of the operations that contribute to higher-level 
flows. In particular, flows provide a quantitative answer to the permissiveness challenges faced
by state-of-the-art dynamic monitors that implement the NSU or the PU strategy.
On the other hand, we consider transitive flows of information between a 
source and a sink, called \emph{source-to-sink flows} 
(Section~\ref{sec:S2SFlows}).
Studying flows at this coarse-grained level is worthwhile because 
source-to-sink flows are what security analysts are interested in 
when using an information flow analysis.

The metrics presented in this section measure the prevalence of flows 
quantitatively, and do not attempt to judge the importance of flows.
To ensure that our flows represent relevant problems, our study uses 
real-world security problems and policies that capture these issues.


\subsubsection{Micro Flows}
\label{sec:metricsMicroFlows}

To measure how many explicit, observable implicit, and hidden implicit 
flows exist, our monitor increments the counters for these micro
flows as follows.

\paragraph*{Explicit flows}
The approach counts an explicit flow for every assignment event where the 
written value is sensitive but the value that gets overwritten (if any) is not 
sensitive. The rationale is to capture program behavior where sensitive 
information flows to a memory location that  stores insensitive information. 
In contrast, overwriting a sensitive value with another (in)sensitive value 
does not leak any new information, and therefore does not count as an 
explicit flow.
    
 For example, consider this code:
\begin{lstlisting}[mathescape]
var x = 3$^H$; var y = 5$^H$; var z;
x = y; // no explicit flow
z = x; // explicit flow
\end{lstlisting}

\paragraph*{Observable implicit flows}
The approach counts an observable implicit flow for every assignment event that 
happens in a sensitive context and that overwrites an insensitive value.  Similar 
to explicit flows, the rationale is to capture program behavior that writes 
sensitive information to a memory location that stores insensitive information. 
The main difference is that the assignment happens because of a control 
flow decision made based on a sensitive context.  Note that it is irrelevant 
whether the written value is sensitive because the fact that a write happens 
leaks sensitive information. 

For example, consider this code:
\begin{lstlisting}[mathescape]
var x = true$^H$; var y = 3; var z;
if (x)
  y = 5; // observable implicit flow
z = 7;   // no flow
\end{lstlisting}

\paragraph*{Hidden implicit flows}
The approach counts a hidden implicit flow for every execution of an upgrade 
statement of a variable containing insensitive information.  The rationale is to capture assignment events that did not
happen, but that could have happened during the execution if a control flow 
decision that depends on a sensitive value would have been different.

For example, consider this code:
\begin{lstlisting}[mathescape]
var x = false$^H$; var y; var z;
if (x)
  y = 5;        // not executed, no flow
upgrade(y);     // hidden implicit flow 
z = y;          // hidden implicit flow
\end{lstlisting}

\subsubsection{Label Creep}
As mentioned earlier, a common reason for false positives is label
creep. Since measuring false positives would be subject to a given
source-to-sink policy, we focus on measuring the prevalence of the
more general phenomenon of label creep in micro flows.
Recall that 
this concept refers to the fact that information flow analysis may 
quickly label a large portion of all values handled in a program as 
sensitive.
In most of the cases, this leads to an explosion in false positives that in 
turn reduces the usefulness of the analysis. We propose a novel metric called
\emph{Label Creep Ratio} (LCR) to assess how many variables and object fields in memory are 
labeled as sensitive.
\begin{equation*}
\LCR = \frac{\textrm{\# sensitive variables/fields ever assigned}}{\textrm{\# variables/fields ever
assigned}}
\end{equation*}
For a given monitoring strategy, the Label Creep Ratio is the ratio between 
the number of assignments of
$H$-labeled values and the total number of assignments.
Intuitively, measuring the LCR throughout an execution
estimates the speed at which the memory locations get assigned sensitive
labels.

\subsubsection{Source-to-sink Flows}
\label{sec:S2SFlows}

To what degree do different kinds of flows contribute to policy violations?
To address this question, we consider transitive flows from a source of 
sensitive information to a sink of insensitive information. For instance, none of the 
flows in the examples above correspond to a source-to-sink flow, since no 
sink statement is present. 

Now, consider the code:

\begin{lstlisting}[mathescape]
var x = false$^H$; var y; var z;
if (x)
  y = 5;
upgrade(y); // hidden micro flow
z = x;      // explicit micro flow
sink(y);    // source-to-sink flow
\end{lstlisting}

The program contains two micro flows and one source-to-sink flow.
However, if the execution is analyzed with taint tracking or observable 
tracking, the source-to-sink flow is missed, because it occurs only due to 
the upgrade statement.

As another example, consider the following code:

\begin{lstlisting}[mathescape]
var x = true$^H$; var y; var z;
if (x)
  y = 5;   // observable flow
z = x;     // explicit flow
sink(y+z); // source-to-sink flow
\end{lstlisting}

The source-to-sink flow will be detected by all three kinds 
of monitoring strategies, because the variable \code{z} gets labeled $H$ via 
an explicit micro flow and then gets passed to the sink.

As illustrated by these two examples, we measure how many source-to-sink 
flows different monitoring strategies detect by tracking what micro flows 
contribute to a source-to-sink flow.
Furthermore, to count the number of unique source-to-sink flows that a 
monitor detects, we compute the set of source code locations involved in 
each source-to-sink flow.
If the code locations of two source-to-sink flows are the same, we 
count them as only one unique flow. This corresponds to the way a human
security analyst would inspect warnings produced by an  
analysis.

\subsubsection{Inference of Upgrade Statements}
\label{sec:onlinePhase}

The approach described so far requires a program that indicates hidden implicit
flows through upgrade statements. To obtain such a program, we adapt a
testing-based technique for automatically inserting upgrade
statements~\cite{birgisson2012boosting}. The basic idea is to repeatedly execute
the program with a particular policy, to monitor the execution for potentially
missed hidden implicit flows (using the PU 
strategy~\cite{austin2010permissive}, see Section~\ref{sec:background1}), and to
insert upgrade statements that signal them to the monitor when counting
micro flows. Whenever the monitor
terminates the program because it detects an access to a value $u$ that is
marked as partially leaked, the approach modifies the program by inserting an
upgrade statement at the code location where $u$ is next used; this upgrade
statement in the modified program will then be executed whenever $u$ is used
again, regardless of whether the same branch that leads to the insertion of the
upgrade statement is taken. The process continues until it reaches a fixed
point, i.e., until the program has enough upgrade statements for the given
tests.

The ability of our analysis to observe hidden implicit flows depends on the
completeness of the inferred upgrade statements, since missing upgrade
statements may result in false negatives for hidden implicit flows. How often
this occurs depends on how well the analyzed executions cover the branches of
the programs. One way to assess this ability would be to measure tradition
branch coverage, i.e., the percentage of all branches that are covered by the
given test inputs. However, traditional branch coverage is only of limited use
because inserting upgrade statements does not rely on covering all branches in
the code, but only on a subset. Specifically, the ability to insert upgrade
statements depends on the branch coverage for conditionals that depend
on sensitive values. We present a metric called \emph{Sensitive Branch Coverage}
(SBC) that captures this idea:
\begin{equation*}
  \SBC = \frac{|\{ c \in C \ \text{where both true and false branch covered}
  \}|}{|C|}
\end{equation*}

where $C$ is the set of conditionals that depend on a sensitive value.
For example, consider executing the following program with 
\code{x=false$^H$}:
\begin{lstlisting}[mathescape]
var x; var y
if (x)/*#\label{line:uncovCond}#*/
  y = 5;
\end{lstlisting}
The set $C$ consists of the conditional at line~\ref{line:uncovCond}, but 
since the execution covers only the false branch, $\SBC =
\frac{0}{1}=0$.

\subsection{Formalization of Flows and Conditions}
{
  \makeatletter
\def\old@comma{,}
\catcode`\,=13
\def,{%
  \ifmmode%
    \old@comma\discretionary{}{}{}%
  \else%
    \old@comma%
  \fi%
}
\makeatother
\label{sec:secFram}
%
We define the syntax and semantics of NanoJS, a simplified
core of JavaScript to illustrate the flow counting performed by
our implementation.

\textbf{Notation: }
We denote empty sequences by $\emptyTrace$. Concatenating two
sequences $\tau_1$ and $\tau_2$ is denoted by $\tau_1 . \tau_2$. Slightly
abusing notation, we also use the same notation to prepend a single element
$\alpha$ to a sequence $\tau$ by writing $\alpha . \tau$. Similarly, we
write $\alpha \in \tau$ to denote that $\alpha$ occurs in sequence $\tau$.

\textbf{NanoJS syntax:}
NanoJS statements:

{\small
\begin{mathpar}
\Stmt \mathrel{::=}
\begin{array}[t]{l}
                 \pskip
  \mathrel{|}    \terminated
  \mathrel{|}    c_1 ; c_2
  \mathrel{|}    \psink{e}
  \mathrel{|}    \passign{x}{e}
  \mathrel{|}    \passignArr{x}{y}{e}
  \mathrel{|}
\\               \pif{e}{c_1}{c_2}
  \mathrel{|}    \pwhile{e}{c}
\end{array} \\

\text{where } x, y \in \Name, \text{ and } e \in \Expr
\end{mathpar}
}

A terminated execution is denoted by $\terminated$. All function calls
to sinks with expression $e$ are modeled by $\psink{e}$; other function
calls are not considered in NanoJS.

\textbf{Semantics: } Operationally, the constructs in NanoJS behave as in
standard imperative languages.
To count micro flows, we associate each primitive value
with a tuple $\cnt : \Cnt$ of flow counts, where $\Cnt = \mathbb{N}^3$.  A 
tuple $(e, o, h) \in \Cnt$ denotes $e$ explicit flows, $o$ observable flows, 
and $h$ hidden flows.
A value is either a primitive value annotated with a flow count, or an 
address on
the heap. We assume that there is a set $\Base$ of primitive base types, 
such as boolean, numbers, and strings. A heap object $o \in \Obj$ maps a 
finite set of names to values. We write \ptrue{}
for boolean value \emph{true} and \pfalse{} for boolean value \emph{false}.

We use flow counts to track how information is propagated by a program,
analogous to labels in other information flow monitors. We define a 
join-semilattice
structure for flow counts as follows.
Intuitively, a non-zero flow count
indicates a sensitive value, whereas if all flow counts are zero, the value 
is insensitive:
The join of two flow counts is defined
as $\cnt_1 \sqcup \cnt_2 = \cnt_1 + \cnt_2$, where $\cnt_1 + \cnt_2$
denotes the pointwise addition of the two flow counts.
Two flow counts satisfy $\cnt_1 \flowsto \cnt_2$
if whenever $\cnt_2 = (0,0,0)$ then $\cnt_1 = (0,0,0)$.

A configuration $\conf{c, \env, h, t, \cnt}$ consists of a statement $c \in
\Stmt$, an environment $\env : \Name \to \Value$ mapping variable names to
values, a heap $h : \Addr \to \Obj$, a stack of security levels $t
\in \Label^\star$, and a sink counter $\cnt : \Cnt$ counting flows reaching sink
statements; we denote the set of configurations by $\Conf$.
An execution of a NanoJS program yields a trace $\Trace = \Value^\star$
indicating outputs produced by the execution.


We now define the small-step semantics of NanoJS. A step $\conf{c, \env, h, t,
  \cnt} \step{\tau} \conf{c', \env', h', t', \cnt'}$ denotes a single
evaluation step producing trace $\tau$.
%
%
We write $\terminated$ for a terminated execution.
Slightly abusing notation, we define $\bigsqcup(h, v)$ as the join of all labels
occurring in value $v$ with heap $h$. For simplicity, we assume that there 
are no cyclical references on the heap.

The function $\upgrade(x, \env, h)$ denotes the pair $(\env', h')$, where
the hidden flow count of all components of the value of a variable $x \in 
\Name$ is incremented by 1.
%
To update flow counts, we use an auxiliary function $\Delta : \Cnt \times \Cnt
\times \Cnt^\star \to (\Cnt)$.
Intuitively, $\Delta(\cnt_{old}, \ell_{new},
t)$ increments the explicit and observable flow counters for assigning a 
value with flow count $\cnt_{new}$ to a location with label $\cnt_{old}$ 
while the security stack is $t$. We define $\Delta(\cnt_{old}, \cnt_{new}, t) =
(\Delta_e, \Delta_o, 0)$ where $\Delta_e =
\begin{cases}  1 & \cnt_{new} \not\flowsto \cnt_{old} \\
    0 & \text{otherwise}
  \end{cases}$ and $\Delta_o = \begin{cases}
    1 & \ell_{old} = \zeroFlow \land \bigsqcup t \neq \zeroFlow \\
    0 & \text{otherwise}
  \end{cases}$

To define observations based on references passed to
sinks, we use a helper function $\toVal(h, v) : \{ \Base \times \Name^\star 
\}$ that, given a value, returns all references to heap objects reachable 
from the value.
We denote evaluating an expression $e$ in environment $\env$ and
heap $h$ by $\expeval{e}{\env, h}$.
The rules propagate flow counts
into the result values; for example, adding two values with one explicit
flow each will result in two explicit flows in the result.
We assume, contrary to real-world JavaScript, that expressions do not have 
side effects.

\begin{figure}[tb]
  \centering
  {\small
  \begin{mathparpagebreakable}
  \inferrule[E-Assign]
  {
    \expeval{x}{\env, h} = v_x \\
     \cnt_x = \bigsqcup(h, v_x) \\
     \expeval{e}{\env, h} = v^{\cnt_e} \\
    \cnt' = \cnt_e + \Delta(\ell_x, \ell_e, t) \\
    v' = v^{\cnt'} \\
    \env' = \env[x \mapsto v']
  }
  {\conf{\passign{x}{e}, \env, h, t, \cnt} \step{}
    \conf{\terminated, \env', h, t, \cnt}}
 \and
    \inferrule[E-If]
  { \expeval{e}{\env, h} = v^\cnt \\
    i = {\begin{cases}
        1 & v = \ptrue\\
        2 & \text{otherwise}
        \end{cases}}
  }
  { \conf{\pif{e}{c_1}{c_2}, \env, t, \cnt} \step{}
    \conf{\pseq{c_i}{\ppop}, \env, \cnt . t', \cnt} }
  \and
  \inferrule[E-Sink]
  { \expeval{e}{\env, h} = v^{\cnt} \\
    \cnt_a = \bigsqcup(v, h) \\
    \cnt' = \cnt + \cnt_a + \Delta(\zeroFlow, \cnt_a, t)
  }
  {\conf{\psink{e}, \env, h, t, \cnt} \step{\toVal(h, v)}
   \conf{\terminated, \env, t, \cnt'}}
  \and
  \inferrule[E-UpgradeL]
  { \expeval{x}{\env, h} = v^{\zeroFlow} \\
    v' = v^{(0, 0, 1)} \\
    (\env', h') = \upgrade(x, \env[x \mapsto v'], h)
  }
  {\conf{\pupgrade{x}, \env, h, t, \cnt} \step{}
    \conf{\terminated, \env', h', t, \cnt}}
\end{mathparpagebreakable}
}
  \caption{Rules for NanoJS with flow counting.}
  \label{fig:rules}
\end{figure}

Finally, Figure~\ref{fig:rules} gives the rules of small-step operational 
semantics for NanoJS with flow counting.
The way the rules modify the environment and heap is standard.
Some standard rules are omitted and provided in the appendix.
In addition to the standard execution of a program, the semantics also track 
flow counts
for each value. For example, an assignment statement $\passign{x}{e}$
propagates the flow counts of the assigned expression $e$ and additionally
increments the explicit flow count if $e$ has non-zero flows and the 
observable flow count if the control-flow path
is determined by sensitive data.
A sink statement $\psink{e}$ increments global counts representing
source-to-sink flows.
Since all sink statements model writes to
\emph{insensitive} sinks, any write of an expression with non-zero
flow counts will result in incrementing the global counters.

\textbf{Security conditions: }
We also adapt existing security conditions
for tracking only explicit or observable flows to NanoJS~\cite{DBLP:conf/esorics/BalliuSS17}. To capture only
explicit flows, we use the notion of \emph{explicit secrecy}; intuitively, a
run of a program satisfies explicit secrecy if and only if the program
obtained by sequentially composing all non-control-flow commands executed
during that run does not leak information. For example, the program
$\pseq{\pif{h}{\passign{l}{1}}{\passign{l}{2}}}{\psink{l}}$ would
produce the extracted programs $\pseq{\passign{l}{1}}{\psink{l}}$
or $\pseq{\passign{l}{2}}{\psink{l}}$ depending on the value
of $h$ in a given run. In both cases, the extracted program
contains prohibited information flows, since the source program
only leaks information through an \emph{implicit} flow.

To track only explicit and observable implicit flows, we keep branching
constructs in the extracted program, but replace not taken branches by $\pskip$.
If the extracted program does not leak sensitive information, then the run satisfies
\emph{observable secrecy}. For example, in the program
$\pseq{\passign{l}{0}}{\pseq{\pif{h}{\passign{l}{1}}{\pskip}}{\psink{l}}}$,
observable secrecy would extract either
$\pseq{\passign{l}{0}}{\pseq{\pif{h}{\passign{l}{1}}{\pskip}}{\psink{l}}}$ or
$\pseq{\passign{l}{0}}{\pseq{\pif{h}{\pskip}{\pskip}}{\psink{l}}}$. This matches
the intuition that an observable flow only occurs in the run where $h$ is
$\ptrue$, where the assignment $\passign{l}{1}$ is executed, but not in
a run where $h$ is $\pfalse$, since this run only leaks information
through a hidden implicit flow; i.e. the extracted program when
$h = \ptrue$ leaks information, but the extracted program for $h = \pfalse$
does not.
Appendix~\ref{app:security-definitions} gives formal definitions of the two 
notions.

\textbf{Soundness: } To establish soundness of our counting scheme, we show that
if all explicit flow counts for all sinks for a given run are $0$, then that run
satisfies explicit secrecy. Similarly, we show that if all explicit and
observable flow counts are $0$, the run satisfies observable secrecy. The formal
theorem statements and proofs can be found in Appendices~\ref{app:soundness}
and~\ref{app:proofs}. 
}

\subsection{Implementation}
\label{sec:implementation}

To implement our methodology, we develop a tool for dynamic information flow analysis 
following Hedin at al.~\cite{Hedin2012,hedin2014jsflow}. The implementation builds on 
Jalangi~\cite{Sen2013}, a dynamic analysis framework for JavaScript that 
uses source-to-source transformation. Since Jalangi supports
ECMAScript~5 only, we down-compile programs written in newer versions of the 
language with Babel~\cite{babel}. Building on top of
Jalangi allows us to focus on the important parts of the analysis and 
let the framework handle otherwise challenging aspects of implementing a 
dynamic information flow analysis, e.g., on the fly instrumentation of code 
produced by \texttt{eval}, exceptional termination of functions, boxing and 
unboxing of primitive values~\cite{Chudnov:2015:IIF:2810103.2813684}.  
We handle higher-order 
functions and track dynamic modification of object properties as described 
by Hedin et al.~\cite{Hedin2012}. Our policy language is 
expressive, allowing the security analyst to mark both functions and 
arguments of callbacks as sources.

To approximate the effects of native calls, we model them by transferring the 
labels from all parameters to the return value. Moreover, if one of the 
parameter is an object, we propagate labels from all its properties to the 
return value. For a set of frequently used native functions, such as 
\texttt{Array.push}, \texttt{Array.forEach}, \texttt{Object.call}, and
\texttt{Object.defineProperty}, we create richer models that propagate 
labels more precisely. To increase the confidence in our implementation, we 
created more than 100 validation tests that assert the correctness of label 
propagation in typical usage scenarios.
When inserting upgrades, the implementation does not modify the actual 
source code but it stores the source code locations of upgrades, and then 
performs the upgrades at runtime.

\section{Empirical Study}\label{sec:results}

This section presents the results of our empirical study that assesses the 
costs and benefits of tracking different kinds of flows.

The last two columns of Table~\ref{tab:apps} show the sensitive branch 
coverage (SBC) and the number of upgrades inserted while executing the 
benchmarks.
Overall, the tests used for the study reach a high SBC, for 54\% of the programs even 100\%, 
enabling the analysis to insert upgrade statements.
For each of the considered benchmarks, our tool can detect source-to-sink 
flows.
This is hardly surprising, since we already know that the programs contain 
such flows, but it shows that our tool can handle complex, real-life 
JavaScript code.




\begin{figure*}
	\centering
	\epsfig{file=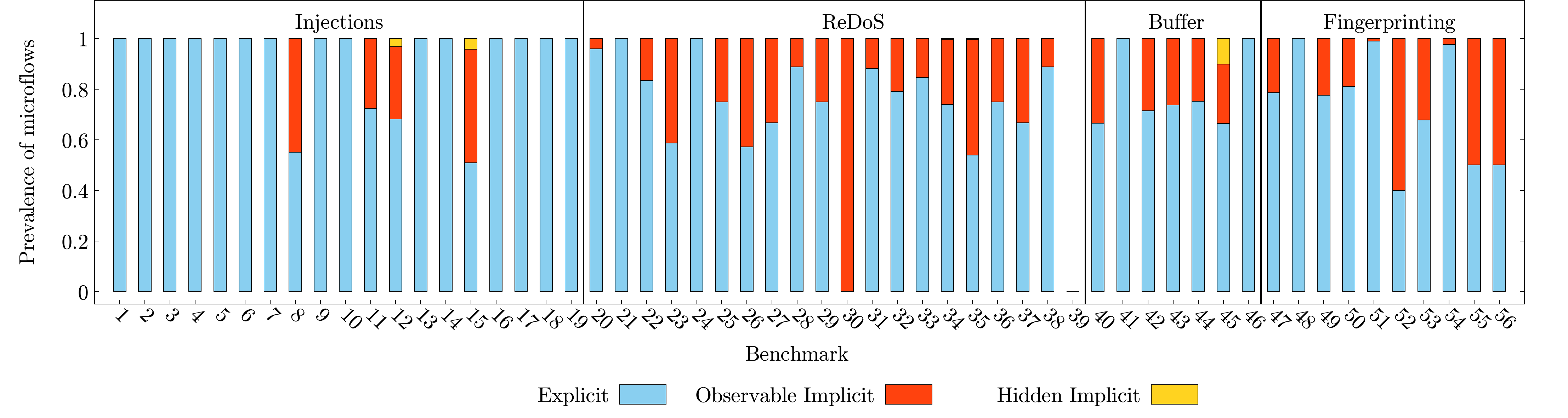, width=\textwidth}
	\caption{Prevalence of micro flows. }
	\label{fig:prevalenceRealistic}
\end{figure*}

\begin{figure*}
	\centering
	\epsfig{file=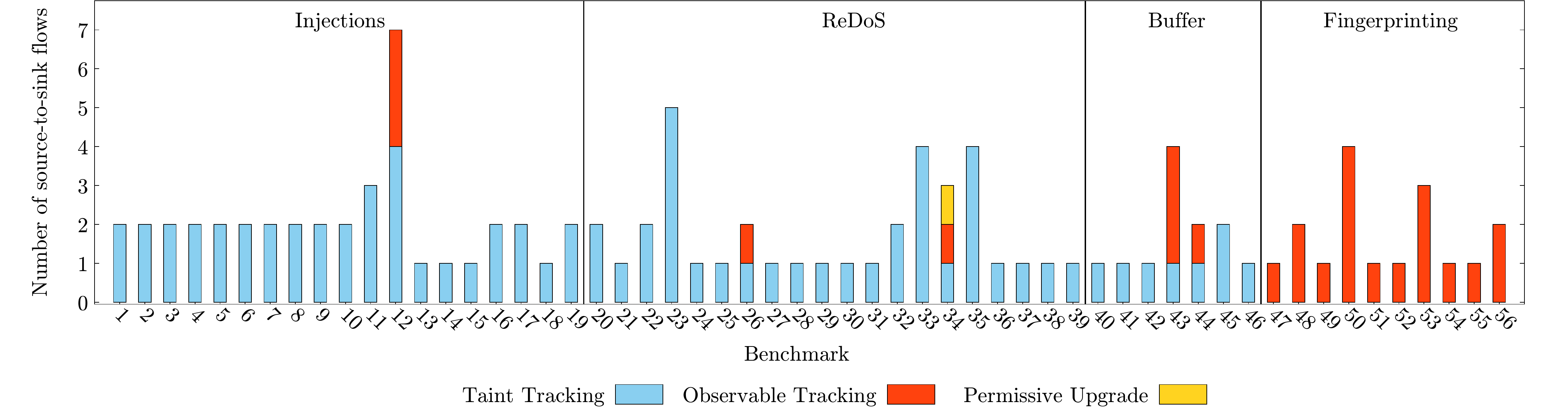, width=\textwidth}
	\caption{Number of source-to-sink flows detected at different security modes. }
	\label{fig:diffLevels}
\end{figure*}

\subsection{Prevalence of Micro Flows}
\label{subsec:mf}

At first, we address the question of how prevalent explicit, observable 
implicit, and hidden implicit micro flows are among all operations that induce 
an information flow.
Figure~\ref{fig:prevalenceRealistic} shows the distribution of micro flows  
for our benchmarks. The majority of benchmarks contain both implicit and 
explicit micro flows. Benchmark~39 is a special case where reaching the sink 
is the first operation performed on the untrusted data, and hence the data 
flows directly from source to sink without producing any micro flow. The 
explicit flows are by far the most prevalent, appearing in all but one 
benchmarks.  \benchmarksWithHidden{} benchmarks also contain hidden implicit 
flows, but we can safely conclude that these cases are rare.


\subsection{Source-to-sink Flows}
\label{subsec:s2s}

We now evaluate source-to-sink flows, which are the ultimate measure of 
success for an information flow analysis.
Source-to-sink flows are what a security analysts ultimately cares about: 
how does information from a sensitive source reaches an insensitive sink.  
Information flow analysis has no way to show that such a flow is 
security-relevant, but it is the analyst's job to further inspect the flows 
and decide. In this section, however, we have a different goal and setup: we 
start with a set of known security problems that produce a source-to-sink 
flow and proceed by showing what type of analysis is needed to detect these 
problems.


Our tool can enforce different 
security conditions (cf. Section~\ref{sec:secFram}). For example, if we are 
interested only in explicit and 
observable implicit flows, we can run the tool in observable tracking mode 
and enforce observable secrecy.
Figure~\ref{fig:diffLevels} presents the number of source-to-sink flows 
detected by different monitoring strategies.
All the integrity vulnerabilities can be detected by taint tracking only, 
and all the security violations in our data set can be detected through 
observable tracking. Moreover, all the Node.js vulnerabilities can be 
detected by the taint tracking only, independently of whether they are 
confidentiality or integrity vulnerabilities. We argue that this is because  
our Node.js programs are expected to be trusted. That is, a security issue  
may arise from a programming error, but not by malicious intention.  This 
assumption does not hold, however, for the fingerprinting and history 
sniffing benchmarks, where only observable implicit flows contribute to the 
source-to-sink flows. A second explanation for why the implicit flows are 
prevalent in the browser environment is that there are already a set of 
security mechanisms in the browser that prevent certain type of dangerous 
behavior. For example when fingerprinting the login state using images, an 
attacker cannot directly read the bytes of the image due to same origin 
policy, and hence it relies on measuring its width.

We analyzed in detail the additional source-to-sink flows detected by 
observable tracking for benchmarks 12, 26, 34, 43, and 44, and by PU for 
benchmark~34.  In all these cases the reported flows are false positives, 
since they do not allow an attacker to exploit the respective vulnerability. 
In Section~\ref{subsec:lcr}, we discuss in detail why these false positives 
occur when data is propagated through implicit flows.

Our results indicate that observable tracking is enough to tackle all the 
real-life security problems we consider and that taint tracking suffices for 
all the trusted code. We do not claim that there are no real-life security 
problems beyond observable secrecy, we just do not see any in our data set.  
Moreover, we believe that when strong controls are in place, attackers will 
be motivated to use more sophisticated attacks, possibly though the use of 
hidden implicit flows.  However, tracking these flows is expensive as we 
will see in the remainder of this section.

\subsection{Permissiveness}
\label{perm}
\begin{figure}
	\centering
	\epsfig{file=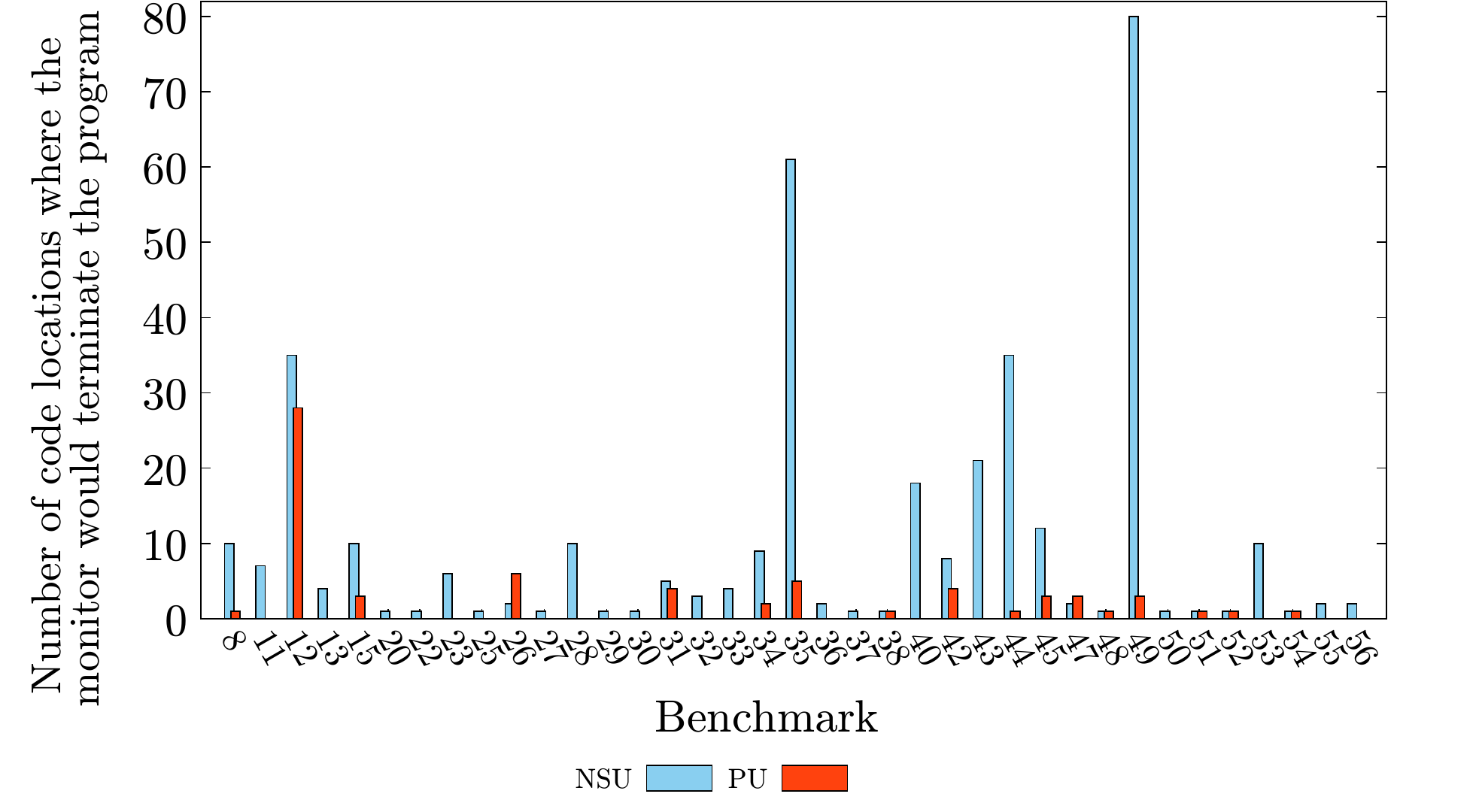, 
		width=0.5\textwidth}
	\caption{Number of violations, summarized by code locations, raised by NSU 
		and PU monitors.}
	\label{fig:permissiv}
\end{figure}

A potential problem for adopting information flow analysis in practice is 
its limited permissiveness, i.e., the fact that a monitor may terminate the 
program even though no data flows from a source to a sink.
Our metrics allow us to quantify this effect both for the NSU and the PU 
monitoring strategies.
Specifically, we measure how many code locations a user would have to 
inspect because a monitor terminates the program.
The NSU monitor terminates the program when an update of an insensitive 
variable is performed in a sensitive context.
This condition corresponds to observable implicit micro flows and we count 
the number of code locations where such a flow occurs.
The PU monitor terminates the program when an insensitive variable that was 
updated in a sensitive context is read.
This termination condition corresponds to the locations where our tool 
inserts an upgrade statement.
Figure~\ref{fig:permissiv} shows the number of code locations affected by 
the lack of permissiveness for NSU and PU.
We exclude benchmarks for which neither of the monitoring strategies raises 
an alarm.
On average, NSU throws \puVsNsuPermissiveness{} times more alarms than PU, 
that is, PU is much more practical than NSU.
However, when comparing the PU violations to the number of source-to-sink 
flows that require PU (Figure~\ref{fig:diffLevels}), we observe that most of 
the PU alarms do not translate to actual source-to-sink flows and should be 
considered false positives.

\subsection{Label Creep Ratio}
\label{subsec:lcr}

\begin{figure}
	\begin{subfigure}[b]{0.5\textwidth}
		\centering
		\epsfig{file=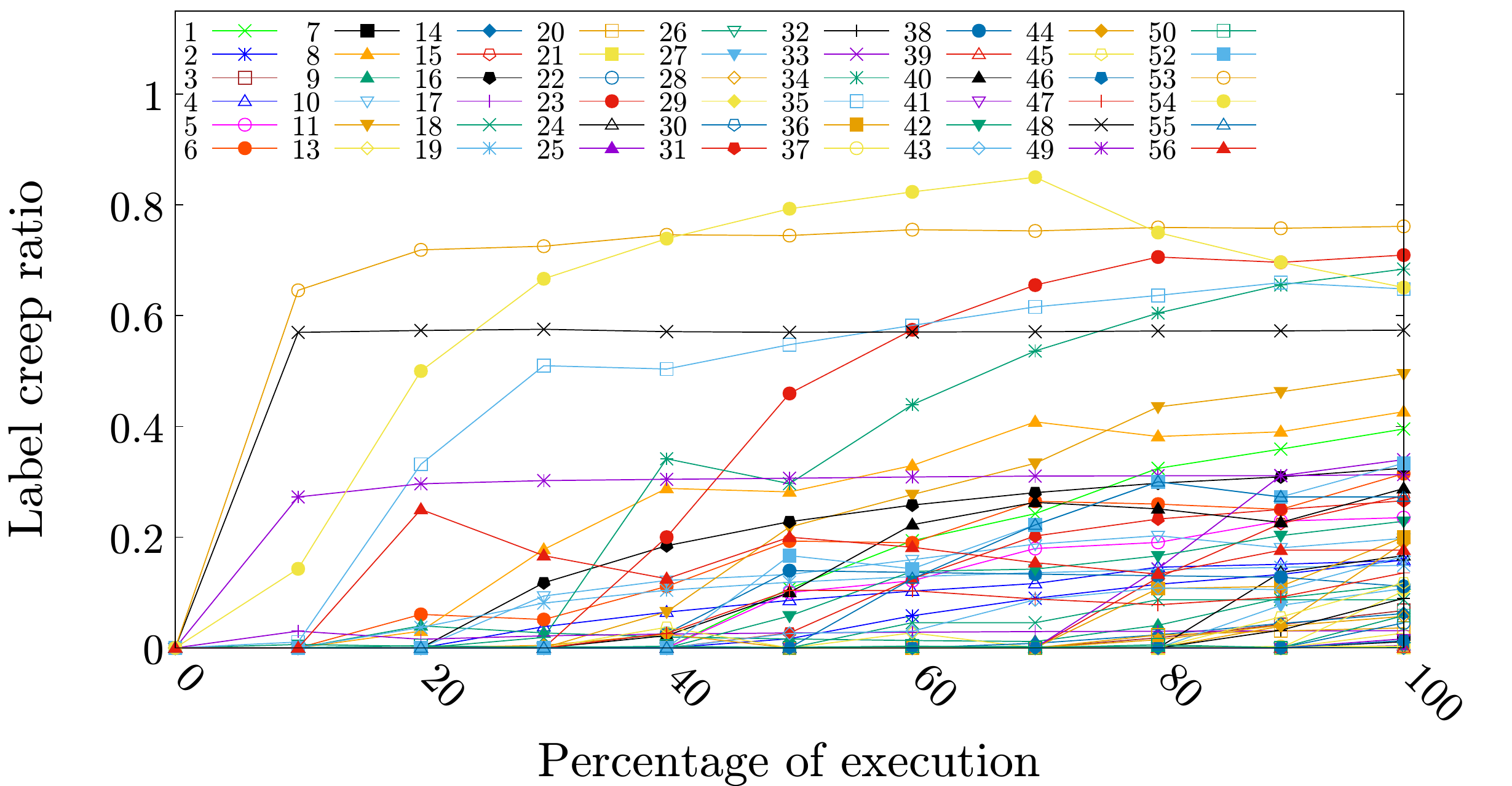, width=\textwidth}
		\caption{Each benchmark}
		\label{fig:labelCreep}
	\end{subfigure}
	\begin{subfigure}[b]{0.5\textwidth}
		\centering
		\epsfig{file=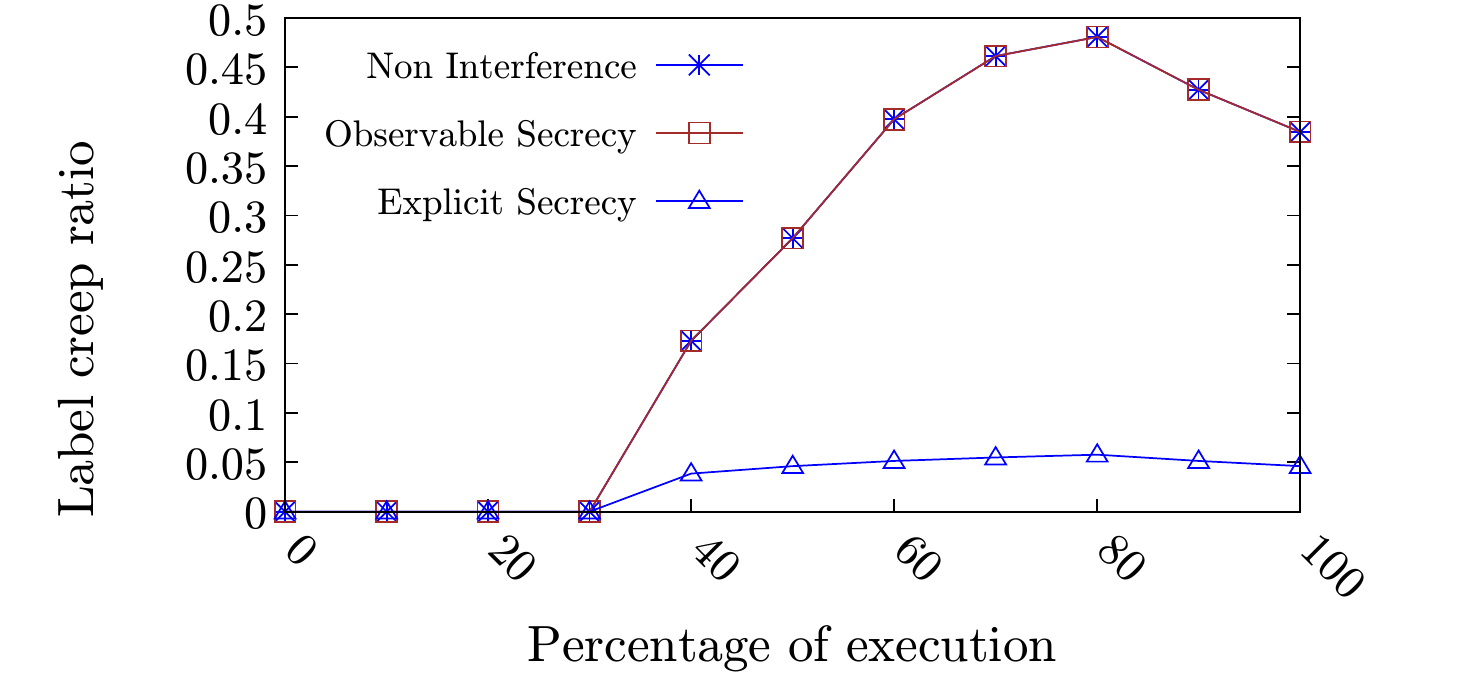, width=\textwidth}
		\caption{Benchmark 11 at various security modes}
		\label{fig:b11case}
	\end{subfigure}
	\vspace{-.2cm}
	\caption{LCR over execution time}
\end{figure}

As a second metric for the cost of different kinds of flows, we 
use the Label Creep Ratio (LCR) defined in 
Section~\ref{sec:metricsMicroFlows}.  For each benchmark and monitoring 
strategy, we measure how the LCR changes during the execution time.
Figure~\ref{fig:labelCreep} shows the ratio for PU monitoring.
The metric is not monotonically increasing because the analysis is 
flow-sensitive, i.e., the security label of a variable may change over time.
Nevertheless, the LCR steadily increases for most benchmarks, which 
confirms the label creep problem.
Because our policies are targeted at detecting known security problems 
in the benchmarks, the maximum LCR reached is relatively low (20\%, on 
average).

A comparison of different monitoring strategies shows that stricter 
monitoring causes more label creep.
On average, observable tracking has a \dropInLCNIOS{} smaller LCR than PU; a 
taint tracking analysis has a \dropInLCOSES{} smaller LCR 
than observable tracking.
Figure~\ref{fig:b11case} illustrates this effect with a representative 
benchmark (number~11).
The graph shows how label creep increases for observable tracking compared 
to taint tracking.

\begin{figure}
	\begin{lstlisting}
// query marked 'sensitive':
function parseQuery(query) {
  // query pushed on the stack:
  if(query instanceof Function) {       
    var nF = eval(query); // sink call
    return [new Part(null, '$', nF)];
  }
}
function Part(f, operator, operand, p){
  if(p === undefined) 
    p = [];                         // implicit
  this.field = f;                   // implicit 
  this.operator = operator;         // implicit
}
	\end{lstlisting}
	\vspace{-0.1cm}
	\caption{Implicit flows snippet from benchmark 11.}
	\label{fig:b11sc}
\end{figure}

We illustrate with the same benchmark~11 how label creep may translate to 
false positives.
By revisiting Figure~\ref{fig:diffLevels}, we observe that the implicit 
flows do not contribute additional source-to-sink violations compared to a 
taint analysis.
Figure~\ref{fig:b11sc} shows an excerpt of the source code of the 
benchmark.
The code is vulnerable to code injection, where \code{query} is the source 
and \code{eval} is the sink.
The source-to-sink flow is trivial since the sensitive data is directly 
passed to the sink at line~3, which a taint tracker easily detects.
In addition, observable tracking pushes the query on the security stack at 
line~2, which causes implicit flows at lines 10 and 11 where two constants 
are written to memory.
For detecting code injections, these flows are irrelevant.
For example, suppose we have a statement \code{eval(this.operator)} at 
line~12, for which observable tracking would report a source-to-sink flow.
This source-to-sink flow would be a false positive because the attacker can 
only control whether the call to \code{eval} happens, not what value flows into it.  

\subsection{Runtime Overhead}
\label{runtime}

\begin{table}[tb]
	\caption{Number of instrumented operations handling sensitive data for 
		different benchmarks and monitors.}
	\centering
	\setlength{\tabcolsep}{1pt}
	\footnotesize
	\renewcommand{\arraystretch}{0.9}
	\begin{tabular}{@{}l|rrr|rrr|rrr@{}}
		\toprule
		& \multicolumn{3}{c}{Explicit Secrecy} & \multicolumn{3}{c}{Observable 
			Secrecy} & 
		\multicolumn{3}{c}{Non Interference}\\
		\midrule
		& Min & Avg & Max & Min & Avg & Max & Min & Avg & Max\\
		\midrule
		Command injection &10 & 59,339 & 1,118,862 & 10 & 59,383 & 1,118,910& 10 & 
		59,540 & 1,118,941\\
		ReDoS vuln. &3 & 210 & 2,064 & 3 & 540 & 6,152& 3 & 633 & 7,073\\
		Buffer vuln. &98 & 5,740 & 24,690 & 98 & 6,007 & 24,748& 98 & 6,084 & 
		24,843\\
		Client-side progr. &4 & 5,919 & 40,364 & 14 & 19,555 & 134,765& 16 & 20,890 & 
		136,502\\
		
		\bottomrule
	\end{tabular}
	\label{fig:costOps}
\end{table}

The last cost metric we use is a proxy measure for the runtime overhead 
imposed by different monitors.
For each benchmark we count the number of operations that propagate a label 
or that modify the security stack.
Table~\ref{fig:costOps} shows how the number of events depends on the kind 
of monitor, aggregated by the different types of vulnerabilities we 
consider.
As expected, raising the security bar translates into larger runtime 
overhead.
Interestingly, this increase is not uniform across the different types of 
benchmarks.
For injections, the cost increase is relatively small, while for ReDoS and 
client-side programs the increase between explicit and observable secrecy is 
more than 2.5-fold.
We hypothesize that this is due to the structure of the programs: when 
comparing these findings with the micro flows in 
Figure~\ref{fig:prevalenceRealistic}, we see that implicit flows are more 
common in ReDoS and client-side programs than in injections.
The price paid to track implicit flows in the client-side benchmarks
translates to detected source-to-sink flows, as we have seen in 
Section~\ref{subsec:s2s}, while this is not the case for ReDoS 
vulnerabilities.

\subsection{Threats to Validity}
The validity of the conclusions drawn from our study is subject to several 
threats.
First, our dynamic information flow analysis uses a simple model for native 
functions (Section~\ref{sec:implementation}), which may not accurately 
capture all effects of these functions. To minimize the influence of this 
limitation, we focus the study on subject programs that have relatively few 
native calls. We also wrote a set of precise models for some of the array 
and string native functions.
Second, our results are limited to the programs we consider and may not 
generalize to other programs or classes of programs. In particular, we mostly 
consider non-malicious programs, where difficult-to-analyze flows may be 
less prevalent than in malicious code. Our methodology is generic enough to 
be easily applied to other programs.
Finally, the hidden implicit flows that our methodology can observe are 
bounded by the upgrade statements inserted into the programs, which in turn 
depend on the tests we use to exercise the programs. To mitigate this threat 
we constructed tests in a way that increases the sensitive branch coverage.  
However, multiple paths cannot be covered due to a variety of reasons, e.g.,  
error cases that cannot be easily triggered or unfeasible execution paths.
Despite these limitations, our study produces interesting insights about the 
kinds of flows that appear in real-world JavaScript programs and the 
cost-benefit tradeoff of information flow analysis.

\section{Related Work}
\label{sec:relatedwork}

\begin{table}
  \caption{JavaScript information flow analyses and the flows they 
  support: \checkmark = considers this flow, 
  \text{\sffamily -} = does not consider this flow,  
  \text{\sffamily NSU} and \text{\sffamily PU} = may
  abort due to NSU or PU checks, respectively,
  and  \text{\sffamily MOD} = may modify program behavior.}
    \footnotesize
  \setlength{\tabcolsep}{8pt}
  \centering
  \renewcommand{\arraystretch}{0.9}
    \begin{tabular}{@{}lllccc@{}}
	\toprule
  Work & Analysis & Explicit &  Obs. & Hidden \\
    \midrule
    Vogt et al.~\cite{DBLP:conf/ndss/VogtNJKKV07} & dynamic & \checkmark & \checkmark  & \text{\sffamily -}  \\
    Jang et al.~\cite{Jang2010} & hybrid & \checkmark & \checkmark  & \text{\sffamily -}  \\    
    Chugh et al.~\cite{chugh2009staged} & hybrid &  \checkmark & \checkmark  & \checkmark  \\
    Tripp et al.~\cite{Tripp:2014:HSA:2610384.2610385} & hybrid & \checkmark & \text{\sffamily -}  & \text{\sffamily -}  \\
    Chudnov \& Naumann~\cite{Chudnov:2015:IIF:2810103.2813684} & dynamic & \checkmark & \checkmark  & \text{\sffamily\footnotesize NSU}  \\
    Hedin et al.~\cite{hedin2014jsflow}& dynamic  &  \checkmark & \checkmark 
    & \text{\sffamily\footnotesize NSU}\\
    Bichhawat et al.~\cite{DBLP:conf/esorics/BichhawatRJGH17} & dynamic &  \checkmark & 
                                                                                        \checkmark & \text{\sffamily\footnotesize PU}\\
       Kerschbaumer et al.~\cite{DBLP:conf/isw/KerschbaumerHLB13} & dynamic &  \checkmark & 
                                                                                        \checkmark & -\\
    Bauer et al.~\cite{bauer2015run} & dynamic &  \checkmark & \text{\sffamily -} & \text{\sffamily -}\\
    De Groef et al.~\cite{DeGroef:2012:FWB:2382196.2382275} & dynamic & \text{\sffamily\footnotesize MOD} & \text{\sffamily\footnotesize MOD} & \text{\sffamily\footnotesize MOD}\\
    Austin \& Flanagan~\cite{Austin:2012:MFD:2103656.2103677} & dynamic & \text{\sffamily\footnotesize MOD} & \text{\sffamily\footnotesize MOD} & \text{\sffamily\footnotesize MOD}\\

	\bottomrule
  \end{tabular}
  \label{tab:recent-work}
\end{table}

%

Denning and Denning pioneered  the development and formal description of
static information flow analyses~\cite{denning1977certification,Denning1976}.
Fenton studies purely dynamic information flow monitors~\cite{DBLP:journals/cj/Fenton74}.
A huge body of work has been created during the years to refine Dennings' and Fenton's 
ideas and to adapt them to various programming languages.
Table~\ref{tab:recent-work} presents some of the more recent tools and 
shows what kinds of flows they consider.
Many analyses consider only explicit flows~\cite{DBLP:conf/sp/SchwartzAB10}.  Among the analyses that 
consider implicit flows, the majority stop or modify the program as soon as a hidden 
flow occurs.

\paragraph*{Information Flow Analysis for JavaScript}
Chugh et al.\ propose a static-dynamic analysis that reports flows from code 
given to \code{eval()} to sensitive locations, such as the location bar of a 
site~\cite{chugh2009staged}.
Austin and Flanagan address the problem of hidden implicit 
flows~\cite{austin2009efficient,austin2010permissive}, as discussed in 
detail in Section~\ref{sec:background}.
Hedin et al.\ propose a dynamic analysis that implements the NSU strategy 
for a subset of JavaScript~\cite{Hedin2012}. They develop
JSFlow, which supports the full JavaScript language, but it requires inserting 
upgrade statements manually~\cite{hedin2014jsflow}.
Birgisson et al.\ propose to  automatically 
insert upgrade statements~\cite{birgisson2012boosting} by iteratively 
executing tests under the NSU monitor. Their approach is implemented for a 
JavaScript-like language, whereas we support the full JavaScript language.  
Our monitor implements the PU strategy to insert upgrade statements,
which reduces the number of upgrade statements and increases permissiveness.  
Bichhawat et al.\ propose a variant of PU, where the program is terminated 
whenever a partially leaked value may flow into the 
heap~\cite{Bichhawat2014}.
A WebKit-based browser by Kerschbaumer et
al.~\cite{DBLP:conf/isw/KerschbaumerHLB13} balances performance and
permissiveness by probabilistically switching between taint tracking
and observable tracking and deploys crowdsourcing techniques to
discover information flow violations by Alexa Top 500 pages.

\paragraph*{Other Work on Information Flow Analysis}
Balliu et al.\ study a family of information flow trackers for different 
kinds of flows and propose security conditions to evaluate their 
soundness~\cite{DBLP:conf/esorics/BalliuSS17}.
We borrow their conditions to prove the soundness of our monitor for NanoJS.  
Bao et al.\ show that considering implicit flows can cause a significant 
amount of false positives and propose a criterion to determine a subset of 
all conditionals to consider~\cite{Bao2010}.
Chandra et al.\ propose a \mbox{VM-based} analysis for Java that combines a 
conservative static analysis with a dynamic analysis to track all three 
kinds of flows considered in this paper~\cite{Chandra2007}.
Dytan is a dynamic information flow analysis for binaries that supports both 
explicit and observable implicit flows~\cite{Clause2007}.
Myers and Liskov introduce Jif, a language for specifying 
and statically enforcing security policies for Java programs~\cite{MyersL00}.
A survey by Sabelfeld and Myers provides an overview of further static 
approaches~\cite{Sabelfeld2003}.

\paragraph*{Applications of Information Flow Analysis}
Information flow analysis is widely used to discover potential 
vulnerabilities. All approaches we are aware of consider only a subset of 
the three kinds of flows.
Flax uses taint analysis to find incomplete or missing input validation and 
generates attacks that try to exploit the potential 
vulnerabilities~\cite{Saxena2010a}.
Lekies et al.~\cite{Lekies2013} and Melicher et al.~\cite{domxss:ndss18}  propose a similar approach to detect DOM-based XSS 
vulnerabilities.
Jang et al.\ analyze various web sites with information flow policies 
targeted at common privacy leaks and attack vectors, such as cookie stealing 
and history sniffing~\cite{Jang2010}. Their analysis considers observable 
implicit flows but not hidden implicit flows.
Sabre analyzes flows inside browser extensions to discover malicious 
extensions~\cite{dhawan2009analyzing}. Their analysis considers only 
explicit flows.

\paragraph*{Studies of Information Flow}
King et al.~\cite{King2008} share our goal of understanding practical
trade-offs between explicit and implicit flows. They empirically study implicit flows detected by a 
static analysis in six Java-based implementations of authentication
and cryptographic functions.
They report that most of the reported policy violations are 
false positives, mostly due to conservative handling of exceptions.
Our work focuses on dynamic analysis for JavaScript-based
implementations, which gives rise to a class of observable secrecy
monitors that is not relevant in a static setting.
Another empirical study of information flows is by Masri et 
al.~\cite{masri2009measuring}. Their work studies how the length of flows 
(measured as the length of the static dependence chain), the strength of 
flows (measured based on entropy and correlations), and different kinds of 
information flows (explicit and observable implicit) relate to each other.
Similar to our methodology, Masri et al.\ target dynamic analysis.
Our work differs by addressing different research questions, a different 
language, and by considering hidden implicit flows.

\section{Conclusions}
\label{sec:conclusions}

This paper presents an empirical study of information flows in real-world 
programs.
Based on novel metrics to capture the prevalence of explicit, observable 
implicit, and hidden implicit flows, as well as the costs they involve, we 
study \numberapps{} JavaScript programs that suffer from real-world security 
problems.
Our results show that implicit flows are expensive to track in terms of 
permissiveness, label creep, and runtime overhead.
We find taint tracking to be sufficient for most of the studied 
security problems, while for some privacy scenarios observable
tracking is needed.
Our work helps security analysts and analysis developers to better 
understand the cost-benefits tradeoffs of information flow analysis.  
Furthermore, our findings highlight the need for future research on 
cost-effective ways to analyze hidden implicit information flows.

\section*{Acknowledgments}
Parts of this work was supported by the German Federal Ministry of
Education and Research and by the Hessian Ministry of Science and the
Arts with the National Research Center for Applied Cybersecurity, and by
the German Research Foundation within the ConcSys and Perf4JS projects.

\bibliographystyle{ACM-Reference-Format}
\bibliography{references}

\appendix
\section{Security Definitions}\label{app:security-definitions}

The previous section has formally defined the flow counting that is at the
heart of our empirical study.
We now related the flow counting to three previously
described~\cite{DBLP:conf/esorics/BalliuSS17} security conditions:
Explicit secrecy, which requires the absence of explicit flows,
observable secrecy, which requires the absence of both explicit
flows and observable hidden flows, and
non-interference, which requires the absence of all three kinds of flows,
i.e., explicit flows, observable implicit flows, and hidden implicit flows.
To describe these security conditions in our formalization,
we define an instrumented version of the semantics that, along with counting
flows, extracts another program.
Intuitively, the extracted program preserves the semantics of the original
program execution but exposes only a subset of all flows.

To formalize non-interference, we first refer to low-equivalence on environments
and heaps. Two environments and heaps are low-equivalent if they are equal on
all insensitive values. For example, when considering
integrity, the two states are equal on all non-attacker-controlled variables.
Dually, for
confidentiality, this
means that the attacker cannot observe any difference between the two states.
Non-interference is defined in terms of low-equivalence of initial environments
and heaps. Intuitively, an execution satisfies non-interference iff the same
trace can be produced for any indistinguishable starting environment and heap.

\begin{mydef}[Non-interference]
  A program $c$ satisfies non-interference for environment $\env_1$ and heap
  $h_1$, iff whenever $\conf{c, \env_1, h_1, [], \zeroFlow}$ $\step{\tau}^\star \conf{c',
    \env_1', h_1', t_1', \cnt_1'}$, then for all $\env_2$, $h_2$, where $(\env_1, h_1)
  =_L (\env_2, h_2)$, it holds that $\conf{c, \env_2, h_2, [], \zeroFlow}
  \step{\tau}^\star \mathit{cnf'}$ for some $\mathit{cnf'}$.
\end{mydef}

Explicit secrecy and observable secrecy are both defined by extracting a simpler
program during the execution of a program which eliminates information flows not
considered by the security condition: Programs extracted by explicit secrecy
contain no control-flow information, whereas programs extracted by observable
secrecy discard statements in untaken branches, thus removing leaks through not
executed statements in other branches in the program.

The following describes the program extraction formally.
We extend each configuration $\cnf$  with an extracted
statement $c_e$, written
$(\cnf, c_e)$, where $\cnf \in Conf$ and $c_e \in A$ with $A$ referring to
the set of extracted statements for a given security condition.
For each security condition, we below define an extraction extraction
function $E: \Conf \times A
\to A$ and then use $E$ in the execution steps of the instrumented
semantics: $(\cnf, c_e) \step{\tau}_E (\cnf, E(\cnf, c_e))$.

\textbf{Explicit secrecy: }  For explicit
secrecy, we disregard control-flow-related statements by defining
an extraction function $\expl$. Intuitively,
the extracted program discards all control-flow decisions that influenced the
current execution and extracts only the straight-line portion of the current
execution. As a result, the extracted program no longer contains any
implicit flows.
The extraction function for explicit secrecy is defined as follows:
{\small
\begin{align*}
  & \expl(\conf{\passign{x}{e}, \dots}, c_e) = \pseq{c_e}{\passign{x}{e}} \\
  & \expl(\conf{\psink{f}, \dots}, c_e) = \pseq{c_e}{\psink{f}} \\
  & \expl(\conf{\passignField{x}{y}{e}, \dots}, c_e) = \pseq{c_e}{\passignField{x}{y}{e}} \\
  & \expl(\conf{\pseq{c_1}{c_2}, \dots}, c_e) = \expl(\conf{c_1, \dots}, c_e) \\
  & \expl(\conf{\dots}, c_e) = c_e
\end{align*} }
Based on this extraction function, we can define explicit secrecy:
\begin{mydef}
  A program $c$ satisfies explicit secrecy for $\env$ and $h$ iff whenever
  $(\conf{c, \env, h, \cnt, S}, \pskip) \step{\tau}^\star_{\expl}
   (\conf{c', env', h', \cnt', S'}, c_e')$, then
 $c_e'$ is non-interfering for environment $\env$ and heap $h$.
\end{mydef}


{\newcommand{\cxt}{\mathit{cxt}}
\textbf{Observable secrecy: }
To define observable secrecy we first define evaluation contexts
to keep track of where in a partially extracted program
the next statement should be placed. The set $\Cxt$ of evaluation
contexts is defined by the following grammar:
\begin{mathpar}
\Cxt \mathrel{::=}
\begin{array}[t]{l}
                 \bullet
  \mathrel{|}    \pseq{\Stmt}{\Cxt}
  \mathrel{|}    \pif{e}{\Cxt}{\pskip}
  \mathrel{|}\\
                 \pif{e}{\pskip}{\Cxt}
\end{array} \\
\vspace{-.6cm}
\end{mathpar}
%
Note that with the exception of $\bullet$, symbolizing a hole
in the context, evaluation contexts are a subset of statements.
We denote replacing $\bullet$ by a statement or context $c$
in a context $\cxt$ by $\cxtInsert{\cxt}{c}$. Note that if $c \in \Stmt$,
then $\cxtInsert{\cxt}{c} \in \Stmt$.

We then define an extraction function $\obs : \Conf \times \Cxt \to \Cxt$
to define observable secrecy:
{\small
\begin{align*}
  &\obs(\conf{\passign{x}{e}, \dots}, \cxt) =
    \cxtInsert{\cxt}{\pseq{\passign{x}{e}}{\bullet}} \\
  &\obs(\conf{\passignField{x}{y}{e}, \dots}, \cxt) =
    \cxtInsert{\cxt}{\pseq{\passignField{x}{y}{e}}{\bullet}} \\
  &\obs(\conf{\psink{f}, \dots}, \cxt) =
   \cxtInsert{\cxt}{\pseq{\psink{f}}{\bullet}} \\
  &\obs(\conf{\pseq{c_1}{c_2}, \dots}, \cxt) =
   \obs(\conf{c_1, \dots}, \cxt) \\
  &\obs(\conf{\pif{e}{c_1}{c_2}, \env, h, \dots}, \cxt) = \\
  &\qquad\quad\begin{cases}
    \cxtInsert{\cxt}{\pif{e}{\bullet}{\pskip}} & \expeval{e}{\env, h} = \ptrue \\
    \cxtInsert{\cxt}{\pif{e}{\pskip}{\bullet}} & \text{otherwise}
  \end{cases} \\
  &\obs(\conf{\ppop, \dots}, \cxt) = \mathit{leaveBranch}(\cxt) \\
  &\obs(\conf{\dots}, \cxt) = \cxt
\end{align*}
}
where $\mathit{leaveBranch}$ denotes shifting the hole in the context outside
of the branch of the surrounding $\textbf{if}$. Note that in programs
not initially containing $\ppop$ statements, any $\ppop$ encountered during
execution delimits a control-flow construct.

}

\section{Soundness}\label{app:soundness}

In this section, we show that if a particular execution results in zero explicit
flows, this execution satisfies explicit secrecy. Similarly, if both observable
and explicit flow counts are zero, the run satisfies observable secrecy.

\begin{theorem}\label{thm:explicit-soundness}
  If $(\conf{c, \env, h, t, \cnt}, \pskip) \step{\tau}_{\expl}^\star
  (\conf{\terminated, \env', h', t', \cnt'}, c_e)$
  and $\forall (l, \cnt) \in \tau.\; \cnt(\eFlow) = 0$, then $c_e$
  satisfies explicit secrecy for $\env$ and $h$.
\end{theorem}

Proofs for the two theorems are provided in Appendix~\ref{app:proofs}.

\begin{theorem}\label{thm:observable-soundness}
  If $(\conf{c, \env, h, t_0, \cnt}, \bullet) \step{\tau}^\star
  (\conf{\terminated, \env', h', t', \cnt'}, c_e)$ and $\forall (l, \cnt)
  \in \tau.\; \cnt(\eFlow) = \cnt(\oFlow) = 0$, then $\cxtInsert{c_e}{\pskip}$
  satisfies observable secrecy for $\env$ and $h$.
\end{theorem}

We omit a similar soundness statement for non-interference, as the
monitor follows the same approach as a traditional permissive-upgrade-based
information flow monitor, under the additional assumption that all required
upgrade statements were inserted during the testing phase.

\section{Proofs and Additional Definitions}\label{app:proofs}
{\renewcommand{\baselinestretch}{0.98}
We formally define values and objects as follows.
The sets are defined by (mutual-)inductively:
{\small
\begin{mathpar}
  \inferrule{n \in \Base \\
    \cnt \in \Cnt
  }
  {n^{\cnt} \in \Value}
  \and
  \inferrule{
    a \in \Addr \\
    \cnt \in \Cnt
  }{a^{\cnt} \in \Value}
  \and
  \inferrule
  {
    V \subset_{\text{fin}} \Name \\
    f \in (V \to \Value)
  }
  {
    f \in \Obj
  }
\end{mathpar}
}
The formal definition of joining the labels of values is given
by $\bigsqcup(h, n^{(\ell, \cnt)}) = \ell$ if  $n \in \Base$
and $\bigsqcup(h, a) = \bigsqcup_{x \in \dom(h(a))}{(\bigsqcup(h, h(x)))}$ if $a \in \Addr$.
The helper function $\toVal(h, v) : \{ \Base \times \Name^\star \}$ is
defined as follows:
$\toVal(h, b) = \{ (b, []) \}$ if $b \in \Base$, and $\toVal(h, r) = \{ (b',
p.[x]) | (b', p) \in \toVal(h, h(x)), x \in \dom(h) \}$

The remaining rules for the operational semantics are the following:
{\small
  \begin{mathparpagebreakable}
  \inferrule[E-Skip]
  {\ }
  {\conf{\pskip, \env, h, t, \cnt} \step{} \conf{\terminated, \env, h, t,
      \cnt}}
  \and
 \inferrule[E-AssignField]
 {
   \expeval{x}{\env, h} = a_x \\
   a_x \in \Addr \\
   o = h(a_x) \\
   y \in \dom(o) \\
   \expeval{x.y}{\env, h} = v_{xy} \\
   \cnt_{xy} = \bigsqcup(h, v_{xy}) \\
   \expeval{e}{\env, h} = v^{\cnt} \\
   \cnt' = \cnt + \Delta(\cnt_{xy}, \cnt, t) \\
   o' = o[y \mapsto v^{\cnt'}] \\
   h' = h[a_x \mapsto o']
 }
 {
   \conf{\passignField{x}{y}{e}, \env, h, t, \cnt} \step{}
   \conf{\terminated, \env, h', t, \cnt}
 }
 \and
  \inferrule[E-While]
  {\ }
  { \conf{\pwhile{e}{c}, \env, h, t, \cnt} \step{}\\
    \conf{\pif{e}{\pseq{c}{\pwhile{e}{c}}}{\terminated}, \env, h, t, \cnt}
  }
  \and
  \inferrule[E-UpgradeH]
  { \expeval{x}{\env, h} = v^{\cnt} \\
    \cnt \neq \zeroFlow }
  {\conf{\pupgrade{x}, \env, h, t, \cnt} \step{}
    \conf{\terminated, \env, h, t, \cnt}}
  \and
  \inferrule[E-Seq]{\conf{c_1, \env, h, t, \cnt} \step{\tau} \conf{c_1', \env',
  h', t', \cnt'} }
  {\conf{\pseq{c_1}{c_2}, \env, h, t, \cnt} \step{\tau} \conf{\pseq{c_1'}{c_2},
  \env', h',
      t', \cnt'}}
  \and
  \inferrule[E-SeqEmpty]{\ }
  { \conf{\pseq{\terminated}{c}, \env, h, t, \cnt} \step{}
    \conf{c, \env, h, t, \cnt}}
\end{mathparpagebreakable}
}


Two environments and heaps $(\env_1, h_1)$ and $(\env_2, h_2)$ are
low-equivalent, written $\env_1 =_L \env_2$ iff

  $\forall x.\; \cnt(\env_1(x)) =
\cnt(\env_2(x)) \land (\cnt(\env_1(x)) = \zeroFlow \Rightarrow \toVal(h,
\env_1(x)) = \toVal(h, \env_2(x)))$, $\dom(h_1) = \dom(h_2)$,
$\forall r \in \dom(h_1).\; \cnt(h_1(r)) = \cnt(h_2(r)) \land
 (\cnt(h_1(r)) = \zeroFlow \Rightarrow \toVal(h_1, h_1(r)) = \toVal(h_2,
 h_2(r))$.

{
In the interest of brevity, we elide lemmas about standard properties of the
evaluation relation in the following proofs.

\begin{proof}[Proof of Theorem~\ref{thm:explicit-soundness}]
We define $\cnt_0 = (0, 0, 0)$.
  We define an safety property $I_E(\env_0, h_0) \subseteq \Stmt$ on extracted programs as
  follows: $c \in I_E(\env_0, h_0) :\Leftrightarrow (\forall \env, h.\;
  (\env, h) =_L (\env_0, h_0)
  \land \conf{c, \env, h, [], \cnt_0}
  \step{\tau}^\star \conf{c', \env', h', t', \cnt'} \land (c' = \psink{e} \lor c' =
  \pseq{\psink{e}}{c_2}) \Rightarrow \bigsqcup(e, h) = \zeroFlow)$. Moreover,
  we note that only straight-line programs are extracted for explicit secrecy
  and such programs trivially preserve low equivalence of environments and
  heaps.

  Additionally, we define the predicate $B(\env_0, h_0, \env_0', h_0') \subseteq
  \Stmt$ where $c_e \in B$ iff whenever
  $(\env_0, h_0) \le (\env, h)$, $(\env_0, h_0) =_L (\env, h)$, and
  $\conf{c_e, \env, h, [], \cnt_0} \step{\tau}^\star \conf{\terminated,
    \env', h', t', \cnt'}$, then $(\env', h') \le (env_0', h_0')$, where
  $(\env, h) \le (\env_0, h_0)$ iff the counter of each value in $\env$ and $h$
  is related by $\flowsto$ to the corresponding counter in $\env_0$ and $h_0$.
  Formally, $(\env_1, h_1) \le (\env_2, h_2)$ iff
  $\forall x\ v_1\ \cnt_1\ v_2\ \cnt_2.\; \env_1(x) = v_1^{\cnt_1} \land
  \env_2(x) = v_2^{\cnt_2} \Rightarrow \cnt_1 \flowsto \cnt_2$,
  $\dom(h_1) = \dom(h_2)$, and
  $\forall a\ v_1\ \cnt_1\ v_2\ \cnt_2.\; h_1(a) = v_1^{\cnt_1} \land
  h_2(a) = v_2^{\cnt_2} \Rightarrow \cnt_1 \flowsto \cnt_2$.

  For the induction to succeed we show the stronger statement
  that whenever $(\conf{\mathlist{c, \env_0, h_0, [], \cnt_0}}, \pskip)
  \step{\tau_0}^\star (\conf{\mathlist{c', \env_0', h_0', [], \cnt_0'}}, c_e')$, then
  $c_e' \in I_E(\env_0, h_0)$ and $c_e \in B(\env_0, h_0, \env_0', h_0')$.

  We prove this by induction on $(\conf{\mathlist{c, \env_0, h_0, [], \cnt_0}}, \pskip)
  \step{\tau_0}^\star (\conf{\mathlist{c', \env_0', h_0', [], \cnt_0'}}, c_e')$. The
  reflexive case is trivial. For the transitive case, assume
  $(\conf{c, \env_0, h_0, [], \cnt_0}, \pskip) \step{\tau_0}
   (\conf{c', \env_0', h_0', t', \cnt_0'}, c_e) \step{\tau_0'}
   (\conf{c'', \env_0'', h_0'', t'', \cnt_0''}, c_e')$. Per the induction
    hypothesis we have that $c_e \in I_E(\env_0, h_0)$ and $c_e \in B(\env_0,
    h_0, \env_0', h_0')$. We show
  that $c_e' \in I_E(\env_0, h_0)$ and $c_e' \in B(\env_0, h_0, \env_0'', h_0'')$
  by induction on $\conf{\mathlist{c', \env_0', h_0', t', \cnt_0'}} \step{\tau_0'}
  \conf{\mathlist{c'', \env_0'', h_0'', t'', \cnt_0''}}$.
  The main interesting
  cases are \textsc{E-Assign}, \textsc{E-AssignField}, and \textsc{E-Sink}.

  Case \textsc{E-Assign}:
  We have that $c_e' = \pseq{c_e}{\passign{x}{e}}$.
  We have that $c_e' \in I_E(\env_0, h_0)$ follows from the fact that
  any sink in $c_e'$ is also reachable in $c_e$, hence the claim follows
  from the induction hypothesis.

  To show $c_e' \in B(\env_0, h_0, \env_0'', h_0'')$, we note that
  if $\conf{c_e', \env, h, [], \cnt_0} \step{\tau_0}^\star \conf{\terminated,
    env'', h'', t'', \cnt''}$ then $\conf{c_e', \env, h, [], \cnt_0}
  \step{\tau_0}^\star \conf{\passign{x}{e}, \env', h', t', \cnt'}$ and
  $\env'' = \env'(x := v_e)$, where $v_e$ is the result of the assigned
  expression with incremented counters. We show that the label
  of $x$ is still bounded by the corresponding label of $x$ in $\env_0''$.
  From the induction hypothesis we have that $(\env', h') \le (\env_0',
  h_0')$. If the new label of $x$ in $(\env_0', h_0')$ is not $\zeroFlow$, then
  the claim follows trivially. If it is low, it follows from the previous
  fact that $x$ also receives $\zeroFlow$ in $\env'', h''$.

  Case \textsc{E-AssignField} is analogous.

  Case \textsc{E-Sink}:
  In this case $c_e' \in B(\env_0, h_0, \env_0'', h_0'')$ follows easily
  from the induction hypothesis as the sink statement does not change the
  environment and heap. $c_e' \in I_E(\env_0, h_0)$  follows trivially
  for sink statements already reachable in $c_e$. Since $c_e' =
  \pseq{c_e}{\psink{e}}$, we need to show that this also holds when
  reaching $c_e'$. Since the explicit flow count is still $0$ in $\cnt_0''$,
  we have that this the label of $e$ in $(\env_0', h_0')$ is $\zeroFlow$.
  Since $(\env', h') \le (\env_0', h_0')$, we have that therefore
  the label of $e$ in $\env', h'$ is also $\zeroFlow$, as desired.

  Clearly whenever $c \in I_E(\env_0, h_0)$, then $c$ satisfies per-run
  non-interference wrt. $\env_0$ and $h_0$ since low equivalence between
  memories is preserved and only low expression reach sinks without increasing
  the counter.
\end{proof}
}

In the following proof sketch, we define $\pi_{1,2}((a, b, c)) = (a, b)$ and
overload $\zeroFlow$ to be any $n$-tuple of $0$s. The $\flowsto$ relation
is generalized similarly to arbitrary tuples.

\begin{proof}[Proof of Theorem~\ref{thm:observable-soundness}]
  For the induction to go through, we show the stronger property that whenever
  $(\conf{c, \env, h, [], \cnt_0}, \bullet) \step{\tau}^\star (\conf{c', \env',
    h', t', \cnt'}, c_e)$, $(\env, h) =_L (\env_2, h_2)$, and $\pi_{1,2}(\cnt')
  = \zeroFlow$, then:
  \begin{itemize}
  \item $\conf{c_e, \env, h, [], \cnt_0} \step{\tau}^\star
    \conf{\pseq{\bullet}{\ppop^{\length{t'}}}, \env', h', t', \cnt_1'}$, and
    $\pi_{1,2}(\cnt_1') = \pi_{1, 2}(\cnt')$
    ($1$) and
  \item There exist $\env_2'$, $h_2'$, $t_2'$, $\cnt_2'$ and $S_2'$ such
    that $(\conf{\mathlist{c_e, \env_2, h_2, [], \cnt_0} \step{\tau}^\star
    \conf{\pseq{\bullet}{\ppop^{\length{t'}}}, \env_2', h_2', t_2', \cnt_2'}}
    \land
    \length{t'} = \length{t_2'}) \lor
    (\conf{c_e, \env_2, h_2, [], \cnt_0} \step{\tau}^\star
     \conf{\terminated, h_2', t_2', \cnt_2'} \land \bigsqcup{t'} =
     \High)$ $(2)$.
  \item Moreover, $(\env', h') =_W (\env_2', h_2')$ $(\env_2', h_2') \le
    (\env', h')$, and $t_2' \preceq t'$. $(3)$
  \end{itemize}

  where $\env_1 =_W \env_2$ holds iff all values labeled low in both
  environments are equal in value and lists of levels $xs$ and $ys$
  satisfy $xs \preceq ys$ if there exists a suffix of length $\length{xs}$
  of $ys$ such that all levels are pairwise related by $\flowsto$. We write
  $c^n$ for $n$ copies of $c$ composed with sequential composition.

  We show this by induction on $(\conf{c, \env, h, [], \cnt_0}, \bullet)
  \step{\tau}^\star
  (\conf{\mathlist{c', \env', h', t', \cnt'}}, c_e)$. The reflexive case is trivial.
  For the transitive case, assume $(\conf{c, \env, h, [], \cnt_0}, \bullet)
  \step{\tau}^\star (\conf{\mathlist{c', \env', h', t', \cnt'}}, c_e) \step{\tau'}
  (\conf{\mathlist{c'', \env'', h'', t'', \cnt''}}, c_e')$. From part $(1)$ of the
  induction hypothesis we have that $\conf{c_e, \env, h, [], \cnt}
  \step{\tau}^\star \conf{\pseq{\bullet}{\ppop^{\length{t'}}}, \env', h', t',
    \cnt'}$ $(1')$; from part $(2)$ we obtain $\env_2', h_2', t_2', \cnt_2'$ for which the corresponding disjunction also holds $(2')$; from $(3)$ we
have $(\env', h') =_W (\env_2', h_2')$, $(\env_2', h_2') \le (\env', h')$, and
$t_2' \preceq t'$. We show $(1-3)$ by induction on $(\conf{c', \env',
  h', t', \cnt'}, c_e) \step{\tau'} (\conf{c'', \env'', h'', t'', \cnt''},
c_e')$; the interesting cases are \textsc{E-Assign}, \textsc{E-AssignField},
\textsc{E-IfTrue/False}, \textsc{E-Sink}, and \textsc{E-Pop}. We refer to
the proof obligations as $(1''')$, $(2''')$, and $(3''')$.

Case \textsc{E-Assign}:
  $(1'')$ follows trivially from the semantics. We have $c_e' =
  \cxtInsert{c_e}{\pseq{\passign{x}{e}}{\bullet}}$, $\env'' = \env'[x \mapsto
  v_e]$, where $v_e$ is the result of evaluating $e$ in $\env', h'$ and
  incrementing the counts appropriately.

  We proceed by case distinction on $(2')$. If $\conf{c_e, \env_2, h_2, [],
    \cnt_0} \step{\tau}^\star \conf{\pseq{\bullet}{\ppop^{\length{t_2'}}}, \env_2',
    h_2', t_2', \cnt_2'}$, then we have that $\conf{\mathlist{c_e', \env_2, h_2, [],
    \cnt_0}} \step{\tau}
  \conf{\pseq{\pseq{\passign{x}{e}}{\bullet}}{\ppop^{\length{t_2'}}}, \env_2',
    h_2', t_2', \cnt_2'} \step{[]} \conf{\pseq{\bullet}{\ppop^{\length{t_2'}}},
    \env_2'', h_2', t_2', \cnt_2'}$ where $\env_2'' = \env_2'[x \mapsto v_e^2]$,
  where $v_e^2$ is the result of evaluating $e$ in $\env_2', h_2'$.
  $\length{t_2'} = \length{t'}$ follows from this case in $(2')$ and the fact
  that assignments do not modify the label stack.

  In the case where $\conf{c_e', \env_2, h_2, [], \cnt_0} \step{\tau}^\star
  \conf{\terminated, \env_2', h_2', t_2', \cnt_2'}$, then we have that
  $\bigsqcup(t') \neq \zeroFlow$. Therefore, we have that the counter of
  $\env''(x)$ is not $\zeroFlow$, therefore $(\env_2'', h_2') =_W (\env'', h'')$
  and $(\env_2'', h_2') \le (\env'', h'')$ follows trivially.

  The other statements follow from the induction hypothesis: If $\bullet$ is not
  reached, then $\cxtInsert{c_e}{c_2}$ matches the evaluation of $c_e$ for any $c_2$.

  Case \textsc{E-AssignField}: Analogous to \textsc{E-Assign}.

  Case \textsc{E-If}: Without loss of generality, we only
  discuss the case where the \textbf{then} branch is taken. In this case, we first
  note that $\env'' = \env'$, $h'' = h'$ and $t'' = \cnt_e . t'$ where
  $\expeval{e}{\env', h'} = v_e^{\cnt_e}$; in particular note that
  $\length{t''} = \length{t'} + 1$. For $(1'')$ note that
  from $(1')$ it follows that $\conf{c_e', \env, h, [], \cnt_0}
  \step{\tau}^\star \conf{\pseq{\pif{e}{\bullet}{\ppop^{\length{t'}}}}, \env',
    h', t', \cnt'} \step{[]} \conf{\pseq{\bullet}{\ppop^{\length{t' + 1}}},
    \env', h', t'', \cnt'}$ as desired.

  We proceed by case distinction on $(2')$. If $\conf{c_e, \env_2, h_2, [],
    \cnt_0} \step{\tau}^\star \conf{\pseq{\bullet}{\ppop^{\length{t_2'}}},
    \env_2', h_2', t_2', \cnt_2'}$ and $\length{t_2'} = \length{t'}$, then we
  have that $\conf{\mathlist{c_e', \env_2, h_2, [], \cnt_0}} \step{\tau}^\star
  \conf{\mathlist{\pseq{\pif{e}{\bullet}{\pskip}}{\ppop^{\length{t'}}}, \env_2', h_2',
    t_2', \cnt_2'}}$. Note that for $\expeval{e}{\env_2', h_2'} = v_2^{\cnt_e^2}$
  and $(\env_2', h_2') \le (\env', h')$ we have that $\cnt_e^2 \flowsto \cnt_e$.

  To show $(2'')$ we proceed by case distinction on $v_2 = \ptrue$. If $v_2 =
  \ptrue$, we have that
  $\conf{\pseq{\pif{e}{\bullet}{\pskip}}{\ppop^{\length{t'}}}, \env_2', h_2',
    t_2'', \cnt_2'} \step{[]} \conf{\pseq{\bullet}{\ppop^{\length{t'} + 1}},
    \env_2', h_2', t_2'', \cnt_2'}$ where $t_2'' = \cnt_e^2 . t_2'$. Since
  $\cnt_e^2 \flowsto \cnt_e$, we have that $t_2'' \preceq t''$; we also have
  that $\length{t_2''} = \length{t''}$ trivially.

  In the case where $v_2 \neq \ptrue$, we have that
  $\conf{\mathlist{\pseq{\pif{e}{\bullet}{\pskip}}{\ppop^{\length{t'}}}}, \env_2', h_2',
    t_2'', \cnt_2'} \step{[]} \conf{\mathlist{\pseq{\pskip}{\ppop^{\length{t'} + 1}},
    \env_2', h_2', t_2'', \cnt_2'}} \step{[]}^\star \conf{\terminated, \env_2',
    h_2', t_2''', \cnt_2'}$ where $t_2'''$ is a prefix of $t_2''$ and hence,
  $t_2''' \preceq t''$.

  Note that since $v_2 \neq v_e$ and $(\env_2', h_2') =_W (\env', h')$, we have
  that $\cnt'' \neq \zeroFlow$ as required for this case.

  where $v_e^2$ is the result of evaluating $e$ in $\env_2', h_2'$.
  $\length{t_2'} = \length{t'}$ follows from this case in $(2')$ and the fact
  that assignments do not modify the label stack.

  In the case where $\conf{c_e, \env_2, h_2, [], \cnt_0} \step{\tau}^\star
  \conf{\terminated, \env_2', h_2', t_2', \cnt_2'}$ and
  $\bigsqcup(t') \neq \zeroFlow$, we have
  that then $\conf{c_e', \env_2, h_2, [], \cnt_0} \step{\tau}^{\star}
  \conf{\terminated, \env_2', h_2', t_2', \cnt_2'}$ since $\bullet$ is
  not reached and hence replacing it with $\pif{e}{\bullet}{\pskip}$
  does not affect the execution. Since $t_2' \preceq t'$, we trivially have
  that then also $t_2' \preceq \cnt_e . t'$.
  In both cases, the rest of $(3'')$ follows trivially.

  Case \textsc{E-Sink}: We have that $c_e' =
  \cxtInsert{c_e}{\pseq{\psink{e}}{\bullet}}$,
  $\tau' = [v]$ where $v^{\cnt_v} =
  \expeval{e}{\env', h'}$. Environments and the heap are unchanged. Moreover,
  we have that $(\mathit{fst}(\cnt''), \mathit{snd}(\cnt'')) =
  (\mathit{fst}(\cnt'), \mathit{snd}(\cnt'))$, since the flow counts
  are assumed to be $0$. $(1''')$ follows easily.

  For $(2'')$, note the second alternative of the disjunction of $(2')$ leads
  to a contradiction, since then $t'' \neq \zeroFlow$ and this would imply that
  $\pi_{1,2}(\cnt'') \neq \zeroFlow$, violating the assumption that
  $\pi_{1,2}(\cnt'') = \zeroFlow$.

  We can therefore assume that
  $\conf{\mathlist{c_e, \env_2, h_2, [],
    \cnt_0}} \step{\tau}^\star\allowbreak \conf{\mathlist{\pseq{\bullet}{\ppop^{\length{t_2'}}}, \env_2',
    h_2', t_2', \cnt_2'}}$ and $\length{t_2'} = \length{t'}$.
  Hence we also have that
  $\conf{\mathlist{c_e', \env_2, h_2, [],
    \cnt_0}} \step{\tau}^\star \conf{\mathlist{\pseq{\pseq{\psink{e}}{\bullet}}{\ppop^{\length{t_2'}}}, \env_2',
    h_2', t_2', \cnt_2'}} \step{v_2} \conf{\mathlist{\pseq{\bullet}{\ppop^{\length{t_2'}}},
    \env_2, h_2, t_2, \cnt_2'}}$. Since $\pi_{1, 2}(\cnt'') = \zeroFlow$, we also
  have that the label of $v$ is $\zeroFlow$; from $(\env_2', h_2') \le
  (\env', h')$ we have that $v_2$ is also labeled $\zeroFlow$. With
  $(\env', h') =_W (\env_2', h_2')$ this yields that $v = v_2$, concluding
  $(2'')$.
  $(3'')$ follows easily since heaps, environments, and the label stack are not
  modified by executing a sink statement.

Case \textsc{E-Pop}: By this case we have $c_e' = \leaveBranch{c_e}$.
$(1'')$ follows easily.
For $(2'')$ we again proceed by case distinction on the disjunction
in $(2')$. In the first case, the conclusion follows easily, since we
reach the same state as the execution in $(\env', h')$.
Assume now that $\conf{c_e, \env_2, h_2, [], \cnt_0} \step{\tau}^\star
\conf{\terminated, \env_2', h_2', t_2', \cnt_2'}$ and $\bigsqcup{t'} \neq
\zeroFlow$. We proceed by case distinction on this execution reaching
the branch surrounding $\bullet$ in $c_e$. We denote this branch by
$\pif{e}{c_1}{c_2}$; WLOG assume that $c_2 = \pskip$ and
$c_1 = \pseq{c_1'}{\bullet}$. Then, we show $\conf{c_e', \env_2, h_2,
  [], \cnt_0} \step{\tau}^\star \conf{\pseq{\bullet}{\ppop^{\length{t''}}},
  \env_2', h_2', \mathit{tl}(t_2'), \cnt_2'}$: Since
$c_e$ terminated without reaching $\bullet$, we have that $\expeval{e}{\env_2', h_2'}
\neq \ptrue^{\cnt_2}$ for any $\cnt_2$, and, since $c_2 = \pskip$, we
reach $\pseq{\bullet}{\ppop^{\length{t''}}}$ trivially, concluding the
induction.
  This stronger property then trivially implies non-interference of $\cxtInsert{c_e}{\pskip}$.
\end{proof}
}

\end{document}